\documentclass{scrartcl}

\newif\ifsecs
\secstrue

\newif\ifred
\redfalse

\usepackage[utf8]{inputenc}
\usepackage{latexsym}
\usepackage{amsmath}
\usepackage{amsthm}
\usepackage{amssymb}
\usepackage{graphicx}
\usepackage{ifthen}
\usepackage{calc}
\usepackage{stmaryrd}
\usepackage{enumitem}
\usepackage{hyperref}
\usepackage{alphalph}

\usepackage{tikz}

\DeclareFontFamily{OT1}{pzc}{}
\DeclareFontShape{OT1}{pzc}{m}{it}{<-> s * [1.10] pzcmi7t}{}
\DeclareMathAlphabet{\mathpzc}{OT1}{pzc}{m}{it}

\makeatletter
\renewcommand{\section}{\@startsection
  {section}%
  {1}%
  {0mm}%
  {-1\baselineskip}%
  {0.5\baselineskip}%
  {\normalfont\large\bfseries}%
}
\renewcommand{\subsection}{\@startsection
  {subsection}%
  {2}%
  {0mm}%
  {-1\baselineskip}%
  {0.5\baselineskip}%
  {\normalfont\large\itshape}%
}

\renewcommand{\subsubsection}{\@startsection
  {subsubsection}%
  {3}%
  {0mm}%
  {-1\baselineskip}%
  {0.5\baselineskip}%
  {\normalfont\itshape}%
}

\newsavebox{\tempbox}
\renewcommand{\@makecaption}[2]{
  \vspace{10pt}
  \sbox{\tempbox}{\textbf{#1.} #2}
  \ifthenelse{\lengthtest{\wd\tempbox > \linewidth}}{
    \textbf{#1.} #2\par
  }{
    \begin{center}
      \textbf{#1.} #2
    \end{center}
  }
}

\makeatother

\ifsecs
\numberwithin{equation}{section}

\numberwithin{figure}{section}

\else

\fi

\newtheoremstyle{mythm}
  {}
  {}
  {\itshape}
  {}
  {\bfseries}
  {.}
  {.5em}
  {\thmname{#1}~\thmnumber{#2}\ifthenelse{\equal{\thmnote{#3}}{}}{}{~(\thmnote{#3})}}

\newtheoremstyle{mydefn}
  {}
  {}
  {\upshape}
  {}
  {\bfseries}
  {.}
  {.5em}
  {\thmname{#1}~\thmnumber{#2}\ifthenelse{\equal{\thmnote{#3}}{}}{}{~(\thmnote{#3})}}

\newtheoremstyle{myremark}
  {}
  {}
  {\upshape}
  {}
  {\itshape}
  {.}
  {.5em}
  {\thmname{#1}~\thmnumber{#2}\ifthenelse{\equal{\thmnote{#3}}{}}{}{~(\thmnote{#3})}}

\theoremstyle{mythm}
\ifsecs
\newtheorem{theorem}{Theorem}[section]
\else
\newtheorem{theorem}{Theorem}
\fi
\newtheorem{lemma}[theorem]{Lemma}

\newtheorem{corollary}[theorem]{Corollary}

\newtheorem{conjecture}[theorem]{Conjecture}

\theoremstyle{mydefn}

\newtheorem{example}[theorem]{Example}
\theoremstyle{myremark}
\newtheorem{remark}[theorem]{Remark}
\theoremstyle{mythm}

\newcommand{\uend}{\hfill$\lrcorner$}

\newcounter{claimcounter}
\newenvironment{claim}[1][]{
  \renewcommand{\proof}{\smallskip\par\noindent\textit{Proof. }}
  \medskip\par\noindent%
  \ifthenelse{\equal{#1}{}}{%
    \setcounter{claimcounter}{0}\refstepcounter{claimcounter}\textit{Claim~\arabic{claimcounter}.}
  }{%
    \ifthenelse{\equal{#1}{resume}}{%
      \refstepcounter{claimcounter}\textit{Claim~\arabic{claimcounter}.}
    }{%
      \textit{Claim~#1.}
    }
  }
}{
  \par\medskip
}

\newcommand{\case}[1]{\par\medskip\noindent\textit{Case #1: }}

\newenvironment{cs}{
  \begin{description}
    \renewcommand{\case}[1]{\item[\normalfont\itshape\mdseries Case ##1:]}
  }{
  \end{description}
}

\newlist{caselist}{description}{10}
\setlist[caselist]{font=\itshape\mdseries}

\setenumerate[1]{label=(\arabic*)}
\newlist{eroman}{enumerate}{2}
\setlist[eroman,1]{label=(\roman*)}
\setlist[eroman,2]{label=(\alph*)}
\newlist{ealph}{enumerate}{1}
\setlist[ealph]{label=(\Alph*)}

\newcounter{nlistcounter}

\usepackage{color}
\definecolor{blau}{RGB}{0,84,159}
\definecolor{hellblau}{RGB}{142,168,229}
\definecolor{petrol}{RGB}{0,97,101}
\definecolor{tuerkis}{RGB}{0,152,161}
\definecolor{gruen}{RGB}{87,171,39}
\definecolor{maigruen}{RGB}{189,205,0}
\definecolor{gelb}{RGB}{255,237,0}
\definecolor{orange}{RGB}{255,128,0}
\definecolor{magenta}{RGB}{227,0,102}
\definecolor{rot}{RGB}{204,7,30}
\definecolor{bordeaux}{RGB}{161,16,53}
\definecolor{violett}{RGB}{97,33,88}
\definecolor{lila}{RGB}{122,111,172}
\definecolor{grey}{gray}{0.7}
\definecolor{mittelblau}{RGB}{0,128,255}

\ifred
\newcommand{\alert}[1]{\ifmmode{\color{rot}#1}\else{\itshape\color{rot}#1}\fi}
\else
\newcommand{\alert}[1]{\ifmmode{#1}\else{\emph{#1}}\fi}
\fi

\newcommand{\bigmid}{\mathrel{\big|}}
\newcommand{\Bigmid}{\mathrel{\Big|}}

\newcommand{\angles}[1]{\left\langle#1\right\rangle}

\renewcommand{\hat}{\widehat}

\renewcommand{\vec}[1]{\boldsymbol{#1}}

\newcommand{\into}{\hookrightarrow}
\newcommand{\onto}{\twoheadrightarrow}

\DeclareMathOperator{\rg}{rg}

\DeclareMathOperator{\supp}{supp}

\newcommand{\Fraisse}{Fra\"{\i}ss{\'e}}

\renewcommand{\phi}{\varphi}
\renewcommand{\epsilon}{\varepsilon}

\newcommand{\Nat}{{\mathbb N}}
\newcommand{\Real}{{\mathbb R}}

\newcommand{\Rat}{{\mathbb Q}}

\newcommand{\LC}{\textsf{\upshape C}}
\newcommand{\FO}{\textsf{\upshape FO}}

\newcommand{\CD}{{\mathcal D}}

\newcommand{\CF}{{\mathcal F}}
\newcommand{\CG}{{\mathcal G}}

\newcommand{\CS}{{\mathcal S}}
\newcommand{\CT}{{\mathcal T}}

\usepackage[textwidth=2cm]{todonotes}

\usepackage{algorithmic}

\algsetup{linenodelimiter=.}

\usepackage{float}
\newfloat{algorithm}{ht}{alg}
\floatname{algorithm}{Algorithm}

\renewcommand{\hom}{\textsf{\upshape hom}}
\newcommand{\HOM}{\textsf{\upshape HOM}}
\newcommand{\emb}{\textsf{\upshape emb}}
\newcommand{\epi}{\textsf{\upshape epi}}
\newcommand{\sepi}{\textsf{\upshape s-epi}}
\newcommand{\pihom}{\textsf{\upshape pi-hom}}
\newcommand{\pphom}{\textsf{\upshape pp-hom}}

\newcommand{\TD}{{\mathcal{TD}}}
\newcommand{\conn}[1]{{#1}^{\mathpzc c}}
\newcommand{\cTD}{\conn{\TD}}

\newcommand{\dagle}{\preceq}
\newcommand{\dagsle}{\prec}
\newcommand{\Cols}{\Gamma}
\newcommand{\col}{\gamma}

\DeclareMathOperator{\qr}{qr}
\newcommand{\colres}{\sqsubseteq_\col}
\newcommand{\eqc}[1]{\equiv^{\LC}_{#1}}

\begin{document}
\title{Counting Bounded Tree Depth Homomorphisms}
\author{Martin Grohe\\\normalsize RWTH Aachen University}

\date{}

\maketitle

\begin{abstract}
  We prove that graphs $G,G'$ satisfy the same sentences of
  first-order logic with counting of quantifier rank at most $k$ if
  and only if they are homomorphism-indistinguishable over the class
  of all graphs of tree depth at most $k$. Here $G,G'$ are
  \emph{homomorphism-indistinguishable} over a class $\CF$ of graphs
  if for each graph $F\in\CF$, the number of homomorphisms from $F$ to
  $G$ equals the number of homomorphisms from $F$ to $G'$.
\end{abstract}

\section{Introduction}
Structural information is captured very well by homomorphism
counts. Indeed, an old theorem due to Lov\'asz \cite{lov67} states
that two graphs $G,G'$ are isomorphic if and only if $\hom(F,G)=\hom(F,G')$
for all graphs $F$. Here $\hom(F,G)$ denotes the number of
homomorphisms from graph $F$ to graph $G$; homomorphisms are mappings
between vertices that preserve adjacency. This simple theorem is quite
useful and can be seen as a the starting point for the theory of graph
limits \cite{borchalov+06,lov12,lovsze06}: by associating each graph $G$ with the
vector $\HOM(G):=\big(\hom(F,G)\bigmid F\text{ graph}\big)$, we map
graphs into an infinite dimensional real vector space, which can be
turned into a Hilbert space by defining a suitable inner product. This
transformation enables us to analyse graphs with methods of linear
algebra and functional analysis and, for example, to consider convergent
sequences of graphs and their limits, called \emph{graphons} (see
\cite{lov12}). Vector embeddings of graphs are also crucial for
applying machine learning methods to graphs. Notably, there is a close
connection between homomorphism counts and so-called graph
  kernels
(e.g.~\cite{sheschlee+11,krijohmor19}) and
graph neural networks (e.g.~\cite{maent19,morritfey+19}).

However, not only
the full homomorphism vector $\HOM(G)$ of a graph $G$, but also its
projections on natural subspaces capture very interesting
information about $G$. For a class $\CF$ of graphs, we consider
the projection
\[\HOM_{\CF}(G):=\big(\hom(F,G)\bigmid F\in\CF\big)\] of $\HOM(G)$ onto
the subspace indexed by the graphs in $\CF$. Following
\cite{bokchegrorat19},  we call graphs $G,G'$
\emph{homomorphism-indistinguishable over $\CF$} if
$\HOM_{\CF}(G)=\HOM_{\CF}(G')$. Dvor{\'a}k~\cite{dvo10} proved that two graphs are
homomorphism-indistinguishable over the class $\CT_k$ of graphs of
tree width at most $k$ if and only if they are not
distinguishable by the \emph{$k$-dimensional Weisfeiler-Leman}
algorithm, a well-known combinatorial isomorphism test. As we can
always restrict homomorphism vectors to connected graphs without loss
of information, this implies that  two graphs are
homomorphism-indistinguishable over the class $\CT$ of trees 
 if and only if they are not
distinguishable by the $1$-dimensional Weisfeiler-Leman
algorithm, which is also known as \emph{colour refinement} and
\emph{naive vertex classification}. Via well-known characterisations of
Weisfeiler-Leman indistinguishability in terms of the solvability of
certain natural systems of linear inequalities
\cite{atsman13,groott15,mal14} or systems of
polynomial equations or inequalities
\cite{atsoch18,bergro15,gragropagpak19}, this also yields algebraic
characterisations of homomorphism indistinguishability over classes of
bounded tree width. A related algebraic characterisation was
obtained for homomorphism indistinguishability over the class of paths
\cite{delgrorat18}. It is well-known (though usually phrased
differently) that two graphs are homomorphism-indistinguishable over
the class of cycles if and only if they are \emph{co-spectral}, that is,
their adjacency matrices have the same eigenvalues with the same multiplicities. Böker~\cite{bok18} proved that two graphs are 
homomorphism-indistinguishable over the class of bipartite graphs if
and only if they have isomorphic bipartite double covers. The most
recent addition to this picture is a result due to Man\v cinska and
Roberson~\cite{manrob19} stating that two graphs are
homomorphism-indistinguishable over the class of all planar graphs if
and only if they are \emph{quantum isomorphic}. Quantum isomorphism,
introduced in \cite{atsmanrob+19}, is a complicated notion
that is based on similar systems of equations as those characterising
homomorphism indistinguishability over graphs of bounded tree width,
but with non-commutative variables ranging over the elements of
some $C^*$-algebra.

What we see emerging is a rich theory connecting combinatorics,
structural graph theory, and algebraic graph theory. It turns out that
logic is also an integral part of this theory, not only because some
of the algebraic characterisations of homomorphism
indistinguishability can be phrased in terms of propositional proof
complexity \cite{atsoch18,bergro15,gragropagpak19}, but also because
there is a well-known characterisation of the Weisfeiler-Leman
algorithm and hence homomorphism indistinguishability over classes of
bounded tree width in terms of logical equivalence.  The logic $\LC$
is the extension of first-order logic by counting quantifiers of the
form $\exists^{\ge p}x$ (``there exists at least $p$ elements
$x$''). Every $\LC$-formula is equivalent to a formula of plain
first-order logic. However, we are mainly interested in fragments of
the logic obtained by restricting the quantifier rank or the number of
variables of formulas, and the translation from $\LC$ to first-order
logic preserves neither the quantifier rank nor the number of
variables (see Remark~\ref{rem:C2FO}). The logic $\LC$ and its finite variable
fragments have first been considered by Immerman in the 1980s
\cite{imm87a,immlan90}, and they have played an important role in
finite model theory since then. Cai, Fürer, and Immerman
\cite{caifurimm92} showed that equivalence in the $(k+1)$-variable
fragment $\LC^{k+1}$ of $\LC$ corresponds to indistinguishability by
the $k$-dimensional Weisfeiler-Leman algorithm. Thus, two graphs are
$\LC^{k+1}$-equivalent if and only if they are homomorphism
indistinguishable over the class $\CT_k$ of graphs of tree width at
most $k$.

Rather than restricting the number of variables in a formula,
it is, arguably, even more fundamental to restrict the quantifier rank
(maximum number of nested quantifiers in a formula). Our main result
is the following characterisation of equivalence in the fragment $\LC_k$ of $\LC$
consisting of all formulas of quantifier rank at most $k$.

\begin{theorem}\label{theo:main}
  For all $k\ge 1$ and all graphs $G,G'$ the following are equivalent.
  \begin{eroman}
  \item 
    $G$ and $G'$ are homomorphism-indistinguishable over the class 
    $\TD_k$ of all graphs of tree depth at most $k$. 
  \item
    $G$ and $G'$ satisfy the same $\LC_{k}$-sentences.
  \end{eroman}
\end{theorem}

\emph{Tree depth}, introduced by Ne\v
set\v ril
and Ossona de Mendez \cite{nesoss06},  is a structural graph parameter
that has received a lot of attention in recent
years (e.g.~\cite{bantan16,buldaw14,cheflu18,elbjaktan12,elbgrotan16}). Our result adds a characterisation of
homomorphism indistinguishability over classes of bounded tree depth
to the theory of homomorphism indistinguishability sketched above. 

However, our result is also interesting from a purely logical point of
view. It can be seen simultaneously as a \emph{locality} theorem
and as a \emph{quantifier
elimination} theorem. \emph{Locality}, because as noted above, when
considering homomorphism indistinguishability, we can restrict our
attention to connected graphs. Connected graphs of tree depth at
most $k$ are known to have a radius of at most $2^{k-1}-1$ (see \cite{nesoss12}), and
hence their homomorphic images will always be contained in
neighbourhoods of radius at most $2^{k-1}-1$.
This
means that homomorphism indistinguishability over graphs of tree depth
$k$ and thus $\LC_{k}$-equivalence only depend on 
neighbourhoods of radius at most $2^{k-1}-1$. This consequence of our main
theorem was known before~\cite{lib98}, but we believe that our approach sheds some
new light on locality. It should be seen in
the context of other recent and not-so-recent locality results for
counting logics
\cite{lib98,lib99,lib00a,kusschwe17,kusschwe18,schwe19}. Let us remark
(as already noted by Libkin~\cite{lib98}) that the exact choice of a
counting extension of first-order logic is not so important when we
only study equivalence between structures.\footnote{The reason is
  that over a \emph{fixed} finite graph, formulas of other counting
  extensions of first-order logic, such as the logic
  $\textsf{FOCN}(\mathbb P)$ of \cite{kusschwe17}, are equivalent to
  $\LC$-formulas of the same quantifier rank.}  

Our theorem is a \emph{quantifier-elimination} result, because it says
that we can replace the $k$ nested quantifiers of a $\LC_k$-formula,
which may involve alternations between existential and universal
quantifiers, by flat, unnested homomorphism counts. While new in this
context, replacing quantifier alternation by counting is a common
theme in complexity theory, most prominently represented by Toda's
theorem \cite{tod91} that $\textsf P^{\#\textsf P}$ contains the
polynomial hierarchy.

The proof of our theorem is harder than one might expect in view of
the numerous previous results on homomorphism
indistinguishability. The overall structure of the proof is
as follows: in the first step we use linear algebraic techniques
that go back to Lov\'asz \cite{lov67} to show that homomorphism counts
can be expressed by counts of more restrictive structure
preserving mappings.  In the second step, the connection to logic is
established via an Ehrenfeucht-\Fraisse\ game and interpolation
techniques. To carry out the first step, we need to prove the
invertibility of certain homomorphism matrices, which we achieve by a
decomposition into lower-triangular and upper triangular matrices of
full rank. The precise nature of this decomposition is what makes the
proof difficult; we need to go through various intermediate mappings
obeying certain carefully chosen constraints.

The structure of the paper is simple: we prove the theorem and then
discuss some of its consequences.

\section{Preliminaries}

\subsection{Graphs and Homomorphisms}
We always assume graphs to be undirected and vertex-colour\-ed. Thus a
graph is a triple $(V(G),E(G),\gamma^G)$ where $V(G)$ is a finite set,
$E(G)\subseteq\binom{V(G)}{2}$, and $\col^G:V(G)\to\Cols$ for some set
$\Cols$ whose elements we view as ``colours''.\footnote{For
  clarity of the presentation, we decided to focus on undirected
  graphs here. The result can be extended to arbitrary relational
  structures, see Section~\ref{sec:rel} for a brief discussion.} 
The
\alert{order} of a graph is $\alert{|G|}:=|V(G)|$. A graph $G$ is a
\emph{subgraph} of a graph $H$ (we write $G\subseteq H$) if
$V(G)\subseteq V(H)$, $E(G)\subseteq E(H)$, and
$\gamma^G(v)=\gamma^H(v)$ for all $v\in V(G)$.


A \alert{homomorphism} from a graph $F$ to a graph $G$ is a mapping $h:V(F)\to
V(G)$ such that $h(u)h(v)\in E(G)$ for all $uv\in E(F)$ and
$\col^F(u)=\col^G(h(u))$ for all $u\in V(F)$. We
write \alert{$h:F\to G$} to denote that $h$ is a homomorphism from $F$ to $G$.
We denote the number of homomorphism from $F$ to $G$ by
\alert{$\hom(F,G)$}. Graphs $G,G'$ are \alert{homomorphism-indistinguishable}
over a class $\CF$ of graphs if $\hom(F,G)=\hom(F,G')$ for all
$F\in\CF$; otherwise they are \alert{homomorphism-distinguishable}
over $\CF$.

Observe that for a disconnected graph $F$ with connected components
$F_1,\ldots,F_\ell$ and for an arbitrary graph $G$ it holds that
$\hom(F,G)=\prod_{i=1}^\ell\hom(F_i,G)$. This means that if $\CF$ is a
class of graphs such that all connected components of graphs in $\CF$
belong to $\CF$ as well, then graphs $G,G'$ are
homomorphism-indistinguishable over $\CF$ if and only if they are
homomorphism-indistinguishable over the class \alert{$\conn{\CF}$} of
all connected graphs in $\CF$.

 A homomorphism $h:F\to G$ is an
\alert{embedding} (or \alert{monomorphism}) from $F$ to $G$ (we write
\alert{$h:F\into G$}) if it is injective. A homomorphism $h:F\to G$ is an
\alert{epimorphism} from $F$ to $G$ (we write \alert{$h:F\onto
G$}) if $h$ is surjective and for every edge
$vv'\in E(G)$ there is an edge $uu'\in E(F)$ such that $h(u)=v$ and
$h(u')=v'$. (Note that not every surjective homomorphism is an epimorphism.) If $H:F\onto G$ is an epimorphism, then $G$ is a
\alert{homomorphic image} of $F$.
By \alert{$\emb(F,G)$} and \alert{$\epi(F,G)$} we
denote the numbers of embeddings and epimorphisms from $F$
to $G$. 

If $\pi$ is a partial mapping from $V(F)$ to $V(G)$, then by 
\alert{$\hom(F,G;\pi)$} we denote the number of homomorphisms from $F$ to $G$
that extend $\pi$. In particular, for vertices $u\in V(F)$ and $v\in 
V(G)$, by \alert{$\hom(F,G;u\mapsto v)$} we denote the number of homomorphism 
$h:F\to G$ with $h(u)=v$. We use similar notations for embeddings,
epimorphisms, and other types of mappings that we shall introduce later.

\subsection{First-Order Logic with Counting}

To define the syntax of the logic $\LC$, we assume that we have an
infinite supply of variables, which we denote by $x,y,z$ and
variants such as $x',y_1$. Variables range over the vertices of a
graph.  \alert{Atomic formulas} (in the language of graphs) are of the
form $x=y$, $E(x,y)$ (``there is an edge between $x,y$''), and
$\col(x)=c$ for colours $c$ (``$x$ has colour $c$''). \alert{$\LC$-formulas}
are constructed from atomic formulas using negation $\neg \varphi$,
disjunction $(\varphi \vee \psi)$, and counting quantifiers
$\exists^{\geq p} x\phi$, where $p\in \Nat$, $x$ is a variable, and
$\varphi$, $\psi$ are formulas. 

An occurrence of a variable $x$ is \alert{free} in a formula $\varphi$
if it is outside the range of all quantifications $\exists^{\geq p} x$. A
\alert{sentence} is a formula without any free variables.
We often write $\varphi(x_1,\dots,x_\ell)$ to indicate that the free
variables of $\varphi$ are among $x_1,\dots,x_\ell$. (Not all of these
variables are required to appear in $\phi$.) For a formula
$\phi(x_1,\ldots,x_\ell)$, a graph $G$, and vertices
$v_1,\ldots,v_\ell\in V(G)$, we write
$G\models\phi(v_1,\ldots,v_\ell)$ to denote that $G$ satisfies $\phi$
if the variables $x_i$ are interpreted by the vertices $v_i$. We also
write $\phi(\vec x)$ and $\phi(\vec v)$ for tuples $\vec
x=(x_1,\ldots,x_\ell)$, $\vec v=(v_1,\ldots,v_\ell)$.
Now we can define the semantics of the logic $\LC$ inductively in the obvious
way. In particular, for $\phi(y_1,\ldots,y_\ell)=\exists^{\geq
  p}x\psi(x,y_1,\ldots,y_\ell)$ we let
$G\models\phi(w_1,\ldots,w_\ell)$ if there are mutually distinct
$v_1,\ldots,v_p\in V(G)$ such that $G\models\psi(v_i,w_1,\ldots,w_\ell)$ for
all $i\in[p]$.

The \alert{quantifier rank $\qr(\phi)$} of a $\LC$-formula $\phi$ is defined
inductively by letting $\qr(\phi):=0$ for all atomic formulas
$\phi$ and $\qr(\neg\phi):=\qr(\phi)$,
$\qr(\phi\vee\psi):=\max\{\qr(\phi),\qr(\psi)\}$, and
$\qr(\exists^{\ge p}x\phi):=\qr(\phi)+1$. By \alert{$\LC_k$} we denote the
fragment of $\LC$ consisting of all formulas of quantifier rank at
most $k$. Graphs $G,G'$ are \alert{$\LC_k$-equivalent} if
$G\models\phi\iff G'\models\phi$ for all $\LC_k$-sentences $\phi$. We
write $G\eqc k G'$ to denote that $G$ and $G'$ are $\LC_k$-equivalent
We extend this notation to formulas with free variables, writing
$G,\vec v\eqc k G',\vec v'$ for tuples $\vec v\in V(G)^\ell,\vec v'\in
V(G')^\ell$ to denote that for all $\LC_k$-formulas $\phi(\vec x)$ it
holds that $G\models\phi(\vec v)\iff G'\models\phi(\vec v')$.

\begin{remark}\label{rem:C2FO}
  Interpreting the usual existential quantifier $\exists$ as $\exists^{\ge 1}$,
  we can view first-order logic $\FO$ as a fragment of $\LC$. 
  Observe that $\LC$ has the same expressive power as its fragment
  $\FO$, because $\exists^{\ge p}x\phi(x,y_1,\ldots,y_\ell)$ can be
  equivalently expressed as 
  \[
  \exists x_1\ldots\exists x_p\left(\bigwedge_{1\le i<j\le p}\neg
    x_i=x_j\wedge \bigwedge_{1\le i\le p}\phi(x_i,y_1,\ldots,y_\ell)\right).
  \]
  However, this increases the quantifier rank. It is easy to see that
  for every $k\ge 1$, 
  $\LC_k$ is strictly more expressive than the fragment $\FO_k$ of
  first-order logic consisting of all formulas of quantifier rank at
  most $k$. Actually, for every $k$ the $\LC_1$-formula
  $\exists^{\ge k+1}x (x=x)$ is not equivalent to any
  $\FO_{k}$-formula.
\end{remark}

\subsection{The Bijective Pebble 
  Game}
The \alert{bijective pebble game}, introduced by Hella~\cite{hel96},
gives a combinatorial characterisation of equivalence in the logic
$\LC$ and its fragments $\LC_k$. 

Let $G,G'$ be graphs of the same order. The bijective pebble game on $G$ and $G'$ is played
by two players called \alert{Spoiler} and the
\alert{Duplicator}. \alert{Positions} of the game are pairs
$(\vec v,\vec v')$ where $\vec v\in V(G)^k,\vec v'\in V(G')^k$ for
some $k\ge 0$. A \alert{play} of the game consists of a
sequence of \alert{rounds}, starting from some \alert{initial position}
$(\vec v_0,\vec v_0')$, where $\vec v_0=(v_1,\ldots,v_\ell)$ and $\vec
v'_0=(v'_1,\ldots,v'_\ell)$ for some
$\ell\ge 0$. The default initial position is the ``empty position'' $\big((),()\big)$. In
round $i$ of the game, Duplicator chooses a
bijection $f_i:V(G)\to V(G')$. Then Spoiler chooses a $v_{\ell+i}\in
V(G)$, and we let $v'_{\ell+i}:=f_i(v_{\ell+i})$.
The position after round $i$ is
$(\vec v_i,\vec
v_i'):=\big((v_1,\ldots,v_{\ell+i}),(v'_1,\ldots,v'_{\ell+i})\big)$. In
the
\alert{$k$-round game}, the play ends after $k$-rounds, and Duplicator
\alert{wins} the play if $\vec v_k\mapsto\vec v'_k:=(v_i\mapsto
v'_i\mid 1\le i\le k+\ell)$ is a
\alert{local isomorphism} from $G$ to $G'$, that is, for all
$i,j\in[\ell+k]$ the following conditions are satisfied:
\begin{itemize}
\item
$v_i=v_j\iff v'_i=v'_j$;
\item
$v_iv_j\in E(G)\iff
v'_iv'_j\in E(G')$;
\item
$\col^G(v_i)=\col^{G'}(v'_i)$. 
\end{itemize}
If $\vec v_k\mapsto\vec
v'_k$ is not a local isomorphism, then 
Spoiler wins the play.

We can now define \alert{winning strategies} for Spoiler and Duplicator
in the usual way.

The following lemma, which links the bijective pebble game to the
logic $\LC$, is a minor variant of a theorem due to Hella~\cite{hel96}
and of the standard characterisation of first-order logic in terms of
Ehrenfeucht-\Fraisse\ games (see, for example, \cite{ebbflutho94}).

\begin{lemma}\label{lem:bpgame}
  For all $k,\ell\ge 0$, all graphs $G,G'$ of the same order, and all $\vec
  v\in V(G)^\ell,\vec v'\in V(G')^\ell$ the following are equivalent.
  \begin{eroman}
  \item Duplicator has a winning strategy for the $k$-round bijective
    pebble game on $G,G'$ with initial position $(\vec v,\vec v')$.
  \item $G,\vec v\eqc k G',\vec v'$.
  \end{eroman}
\end{lemma}

If we do not specify the initial position of the game, we always
assume it is the empty position $((),(())$. Thus the lemma implies
that Duplicator has a winning strategy for the $k$-round bijective
pebble game on $G,G'$ if and only if $G\equiv_k^\LC G'$.
\subsection{Graphs of Bounded Tree Depth}

It will be convenient in this paper to view trees and forests as partially ordered
sets. A \alert{forest} $S$ is a pair $(V(S),\dagle^S)$ consisting of a (finite) vertex
set $V(S)$ and a partial order $\dagle^S$ on $V(S)$ such that for every $t\in V(S)$ the set $\{u\in V(S)\mid
u\dagle^S t\}$ is a \alert{chain}, that is, its elements are pairwise
comparable. We denote the
strict partial order associated with $\dagle^S$ by $\dagsle^S$.
If $t\dagsle^S u$ and there is no $v\in V(S)$ such that $t\dagsle^Sv$
and $v\dagsle^S u$, then we say that $u$ is a \alert{child} of $t$ and
that $t$ is the \alert{parent} of $u$. This gives us a 
one-to-one correspondence between forests viewed as partially ordered
sets and rooted forests in the usual graph-theoretic sense. 
The $\dagle^S$
minimal elements of $V(S)$ are called the \alert{roots} of $S$. 
The \alert{height} of $S$ is the length $|X|$ of the longest chain $X$ in
$S$. Note that, differing from the standard graph
theoretic definition, we count the number of vertices (and not the
number of edges) on a path from the root to a leaf. In particular, a
forest consisting of roots only has
height $1$.

A forest
$T$ with a unique root is a \alert{tree}. We denote the root of a tree
$T$ by $r^T$.
A \alert{subtree} of a tree
$T$ is a tree $T'$ with $V(T')\subseteq V(T)$ such that $\dagle^{T'}$ is
the restriction of $\dagle^T$ to $V(T')$. Thus a subtree is an induced
substructure that is a tree itself. Observe that a set $U\subseteq V(T)$
induces a subtree of $T$ if and only if $U$ has a unique
$\dagle^T$-minimal element. This notion of subtree does \emph{not}
coincide with the usual graph-theoretic notion of a subtree of a tree. In particular,
elements of a
subtree can be interleaved with elements that do not belong to the
subtree.

An \alert{elimination forest} of a graph $G$ is a forest $S$ such that
$V(S)=V(G)$ and for every edge $uv\in E(G)$, either $u\dagle^Sv$ or
$v\dagle^S u$. If an elimination forest $S$ of $G$ is a tree, we also
call it an \alert{elimination tree} of $G$. The \alert{tree depth} of
a graph $G$ is the minimum $k$ such that $G$ has an elimination forest
of height $k$. We denote the class of all graphs of tree depth at most
$k$ by \alert{$\TD_k$} and the class of all connected graphs in
$\TD_k$ by $\cTD_k$.

\begin{samepage}
  \begin{lemma}[Ne\v
set\v ril
and Ossona de Mendez~\cite{nesoss06}]\label{lem:ind-td}~
  \begin{enumerate}
  \item $\cTD_1$ consists of all $1$-vertex graphs.
  \item For $k\ge 1$, $\cTD_{k+1}$ is the class of all connected graphs $F$ that
    have a vertex $r$ such that all connected components of
    $F\setminus\{r\}$ are in $\cTD_k$.
  \item For all $k\ge 1$, $\TD_k$ is the class of disjoint unions of
    graphs in $\cTD_k$.
\end{enumerate}
\end{lemma}
\end{samepage}
We let $\CD$ be the class of all pairs $(F,T)$ where
$F$ is a graph and $T$ an elimination tree of $F$.
We usually denote elements of $\CD$ by $D$.\footnote{The reader may
  wonder why we chose the letter ``d'' (in $D$ and $\CD$). One reason is that
  it picks up the ``d'' in depth and that
  $\CD$ is close to $\TD$. Or think of ``d'' as standing for ``decomposed
  graph''.}

For $D=(F,T)\in\CD$, we let $F^D:=F$, $T^D:=T$ and $V(D):=V(F)=V(T)$,
$E(D):=E(F)$, $\col^D:=\col^F$, $\dagle^D:=\dagle^T$, and
$r^D:=r^T$. We call $r^D$ the \alert{root} of $D$. The \alert{height}
of $D$ is the height of $T^D$.  We denote the class of all $D\in\CD$ of
height at most $k$ by $\CD_k$. Observe that a connected graph $F$ is
in $\TD_k$ if and only if there is a $D\in\CD_k$ such that
$F^D=F$.

\begin{remark}
  There is a strange asymmetry in the definition of $\CD$: for pairs
  $(F,T)\in\CD$, we require $T$ to be a tree, not an arbitrary forest, but we
  do not require the graph $F$ to be connected. Yet this definition is
  carefully chosen. In particular, if we required $F$ to be connected
  then we would run into difficulties in the proof of
  Lemma~\ref{lem:6}.
\end{remark}

\section{Past-Preserving Homomorphisms}

Let $D\in\CD_k$, and let $G$ be an arbitrary graph. A
\alert{homomorphism} from $D$ to $G$ is simply a homomorphism from
$F^D$ to $G$. We write \alert{$h:D\to G$} to denote that $h$ is a
homomorphism from $D$ to $G$, and we let $\alert{\hom(D,G)}:=\hom(F^D,G)$ be
the number of homomorphisms from $D$ to $G$. A homomorphism $h:D\to G$
is an \alert{epimorphism} (we write \alert{$h:D\onto G$}) if it is an epimorphism from $F^D$ to $G$.

A homomorphism $h:D\to G$ is \alert{past-injective} if for all
$u,v\in V(D)$ with $u\dagsle^D v$ we have $h(u)\neq h(v)$. If in
addition, for all $u,v\in V(D)$ with $u\dagle^D v$ we have
$uv\in E(D)\iff h(u)h(v)\in E(G)$, then $h$ is \alert{past-preserving}. We
denote the number of past-injective homomorphisms from $D$ to $G$ by
\alert{$\pihom(D,G)$} and the number of past-preserving homomorphisms
from $D$ to $G$ by \alert{$\pphom(D,G)$}.  In this section, we shall
prove that we can compute the numbers of past-preserving homomorphisms
to a graph from the numbers of homomorphisms and vice versa. The
difficult first step will be to establish an equivalence between the
numbers of past-injective homomorphisms and
homomorphisms. 

The general strategy for establishing such an equivalence, going back
to Lov\'asz \cite{lov67}, is to establish a linear relationship
between the corresponding counting vectors, in our case the vectors
$\HOM_{\TD_k}(G)=\big(\hom(F,G)\bigmid F\in\TD_k\big)$ and the
corresponding vector of past-injective homomorphism counts and then
show that the matrix relating the two vectors is invertible (this will
happen in Lemma~\ref{lem:2}, Corollary~\ref{cor:1}, and
Lemma~\ref{lem:3}). On the linear algebra side, we shall write the
(infinite) matrix of homomorphism counts as a product of an
upper-triangular matrix with nonzero diagonal entries and a
lower-triangular matrix with nonzero diagonal entries. This
decomposition of the homomorphism matrix corresponds to a
decomposition of homomorphisms. The upper
triangular matrix is obtained by considering some form of injective
homomorphisms, in our case past-injective homomorphisms. The lower
triangular matrix corresponds to suitable surjective homomorphisms, in
our case \emph{shrinking epimorphisms}, to be introduced next. The
reason that we cannot just work with plain injective and surjective
homomorphisms (or rather epimorphisms) is that the homomorphic image of a graph of tree
depth at most $k$ may have larger tree depth than $k$. However, we
shall prove (in Lemma~\ref{lem:1}) that shrinking epimorphisms
preserve tree depth.

Let $D\in\CD_k$, and let $G$ be a graph with $V(G)\subseteq V(D)$ (but not
  necessarily $G\subseteq F^D$). A \alert{shrinking
  homomorphism} from $D$ to $G$ is a homomorphism $h:D\to G$
such that $h(u)\dagle^D u$ for all $u\in V(D)$ and $h$ is \emph{idempotent}, that is, $h(h(u))=h(u)$ for all $u\in
V(D)$. We are mainly interested in shrinking epimorphisms. We denote the number of shrinking epimorphism from $D$ to $G$ by \alert{$\sepi(D,G)$}.
 Note that
if $h$ is a shrinking epimorphism from $D$
to $G$ then $h(v)=v$ for all $v\in V(G)$. Indeed, since
$h$ is surjective, we have $v=h(u)$ for some $u$ and therefore
$h(v)=h(h(u))=h(u)=v$. This implies that for all $v\in V(G)$ we have
$\col^G(v)=\col^D(v)$. 

To simplify the notation, for graphs $F,G$ we write
\alert{$G\colres F$} if $V(G)\subseteq V(F)$ and
$\col^G(v)=\col^F(v)$ for all $v\in V(G)$. For a $D\in\CD$ we write
\alert{$G\colres D$} instead of $G\colres F^D$.


\begin{lemma}\label{lem:1}
  Let $D\in\CD_k$, $G\colres D$, and let $f:D\onto
  G$ be a shrinking epimorphism from $D$ to 
$G$. Then $T^D$ induces a subtree on $V(G)$, and this
  subtree $T^D[V(G)]$ is an elimination tree of $G$ of height at most $k$,
  that is, $(G,T^D[V(G)])\in\CD_k$.
\end{lemma}

\begin{proof}
  We first prove that $T':=T^D[V(G)]$ is a tree of height at most
  $k$. Observe that $f(r^D)=r^D$ and thus $r^D\in V(G)$. Hence $V(G)$ has a
  unique $\dagle^D$-minimal element, and $T'$ is a tree. Clearly, the
  height of $T'$ is at most the height of $T^D$ and hence at most
  $k$.

  It remains to prove that $T'$ is an elimination tree of $G$. Let
  $vv'\in E(G)$. We shall prove that either $v\dagle^{T'} v'$ or
  $v'\dagle^{T'} v$. Since $f$ is an epimorphism, there is an edge
  $uu'\in E(D)$ such that $f(u)=v$ and $f(u')=v'$. Then $v\dagle^D u$
  and $v'\dagle^D u'$. Since $T^D$ is an elimination tree of $F^D$,
  either $u\dagle^D u'$ or $u'\dagle^D u$. Without loss of generality
  we assume $u\dagle^D u'$. Then $v,v'\dagle^D u'$. Since the set
  $\{t\in V(D)\mid t\dagle^D u'\}$ is a chain in the tree $T^D$,
  either $v\dagle^D v'$ or $v'\dagle^D v$. As $\dagle^{T'}$ is the
  restriction of $\dagle^D$ to $V(G)$, this implies that
  $v\dagle^{T'} v'$ or $v'\dagle^{T'} v$.
\end{proof}

\begin{lemma}\label{lem:2}
  Let $D\in\CD_k$, and let $h:D\to H$ be a homomorphism from $D$ to some graph $H$. Then
  there is a graph $G\colres D$, a shrinking
  epimorphism $f:D\onto G$, and a past-injective
  homomorphism $g:(G,T^D[V(G)])\to H$ such that $h=g\circ f$.

  Furthermore, $G$, $f$, and $g$ are unique. That is, if $G'\colres D$
  and $f':D\onto G'$ is a shrinking epimorphism and $g:(G',T^D[V(G')])\to H$ is past-injective 
  such that $h=g'\circ f'$, then $G=G'$, $f=f'$, and $g=g'$.
\end{lemma}

\begin{proof}
  Let $W:=h(V(D))\subseteq V(H)$ be the range of $h$. Then the sets
  $h^{-1}(w)$, for $w\in W$, form a partition of $V(D)$. For every
  $u\in V(D)$, let $f(u)$ be the $\dagle^D$-minimal element in
  $h^{-1}(h(u))\cap \{t\mid t\dagle^D u\}$. There is at most one such
  element because $\{t\mid t\dagle^D u\}$ is a chain. Note that $f$ is idempotent.
Let $G:=f(F^D)$ be the graph
  with vertex set $V(G):=f(V(D))$ and edge set
  $E(G):=\{f(u)f(u')\mid uu'\in E(D)\}$. Then $f:D\onto G$ is a shrinking epimorphism. Hence by Lemma~\ref{lem:1}, the induced
  subtree $T^D[V(G)]$ is an elimination tree of $G$ of height at most $k$.

  For all $u,u'\in V(D)$, if $f(u)=f(u')$ then $h(u)=h(u')$. Thus
  there is a mapping $g:V(G)\to V(H)$ such that $h=g\circ f$. As $h$
  is a homomorphism and $G=f(F)$, the mapping $g$ is a homomorphism
  from $G$ to $H$. Indeed, for every edge $vv'\in E(G)$ there is an
  edge $uu'\in E(D)$ such that $f(u)=v$ and $f(u')=v'$. Then
  $g(v)g(v')=h(u)h(u')\in E(H)$.

To prove that $g$ is
  past-injective, suppose for contradiction that there are
  $v,v'\in V(G)$ such that $v\dagsle^{T^D[V(G)]} v'$ and
  $g(v)=g(v')=:w$. Note that $v\dagsle^{T^D[V(G)]} v'$ implies
  $v\dagsle^D v'$. As $f$
  is the identity on $V(G)\subseteq V(F^D)$, we have 
  $h(v)=h(v')=w$. By the definition of $f$, this means that $v=f(v)$
  and $v'=f(v')$ are $\dagle^D$-minimal elements in $h^{-1}(w)$. Since
  $v\neq v'$, it follows that $v\not\dagle^D v'$. This is a
  contradiction.

  It remains to prove the uniqueness. Let $G'\colres D$ and $f':D\onto G$
  a shrinking epimorphism and $g':(G',T[V(G')])\to H$ a past-injective
  homomorphism such that $h=g'\circ f'$. If $f=f'$ then $G=G'$,
  because $G'=f(F^D)=f'(F^D)=G$, and $g=g'$ because $g\circ f=g'\circ f'$
  and $f,f'$ are surjective. Suppose for contradiction that
  $f'\neq f$. Let $u\in V(D)$ such that $f'(u)\neq f(u)$ and, subject
  to this condition, $u$ is $\dagle^D$-minimal.
  \begin{cs}
    \case1
    $f(u)\neq u$.\\
    Let $u':=f(u)$. Then $u'\dagsle^D u$ and, since $f$ is idempotent,
    $f(u')=u'$. Thus 
    \[
      h(u)=g(f(u))=g(u')=g(f(u'))=h(u').
    \]
    By the minimality of $u$, we have $f'(u')=f(u')=u'$. This implies
    \begin{equation}
      \label{eq:1}
      g'(f'(u))=h(u)=h(u')=g'(f'(u')).
    \end{equation}
    Since $f'(u)\dagle^D u$ and
    $f'(u')=u'\dagle^D u$ and $T^D$ is a tree, either
    $f'(u)\dagle^D f'(u')$ or $f'(u')\dagle^D f'(u)$. Since $g'$ is
    past-injective, by \eqref{eq:1} we have neither $f'(u)\dagsle^D f'(u')$ nor
    $f'(u')\dagsle^D f'(u)$ and thus $f'(u)=f'(u')=u'=f(u)$. This is a contradiction. 
    \case2
    $f(u)=u$.\\
    Let $u':=f'(u)$. Then $u'\dagsle^D u$. Since $f'$ is idempotent, we have
    $f'(u')=f'(f'(u))=f'(u)=u'$ and thus
    $h(u)=g'(f'(u))=g'(u')=g'(f'(u'))=h(u')$. By the minimality of $u$ we have
    $f(u')=f'(u')=u'$. Hence
      $g(u)=g(f(u))=h(u)=h(u')=g(f(u'))=g(u')$, which contradicts $g$
      being past-injective.
      \qedhere
 \end{cs}
\end{proof}

\begin{corollary}\label{cor:1}
  Let $D\in\CD_k$, and let $H$ be a graph. Then
  \begin{equation}
    \label{eq:19}
    \hom(D,H)=\sum_{G\colres D}\sepi(D,G)\cdot\pihom((G,T^D[(V(G)]),H).
  \end{equation}
 \end{corollary}

\begin{corollary}\label{cor:2}
  Let $D\in\CD_k$, and let $H,H'$ be graphs such that for all
  $G\colres D$ with $(G,T^D[V(G)])\in\CD_k$ it holds that
  \[
    \pihom((G,T^D[V(G)]),H)=\pihom((G,T^D[V(G)]),H').
  \]
  Then 
  \[
    \hom(D,H)=\hom(D,H').
  \]
\end{corollary}

\begin{proof}
  Here we use Lemma~\ref{lem:1} to see that we can restrict the sum in
  \eqref{eq:19} to $G$ with $(G,T^D[V(G)])\in\CD_k$.
\end{proof}

\begin{lemma}\label{lem:3}
  Let $(G,T)\in\CD_k$, and let $H,H'$ be graphs such that for all $F\in\TD_k$ with
  $F\colres G$,
  \[
    \hom(F,H)=\hom(F,H').
  \]
  Then  
  \[
    \pihom((G,T),H)=\pihom((G,T),H').
  \]
\end{lemma}

\begin{proof}
  Let $\CG$ be the set of all
  $(G',T')\in\CD_k$ such that $G'\colres G$ and
  $T'=T[V(G')]$. In particular,
  $(G,T)\in\CG$.

  Let $(G_1,T_1),\ldots,(G_m,T_m)$ be an enumeration of
  $\CG$ such that $(G_m,T_m)=(G,T)$ and $|G_i|\le |G_j|$ for $i\le
  j$. Observe that $\sepi((G_i,T_i),G_i)=1$ for all $i$ and that
  $\sepi((G_i,T_i),G_j)>0$ for $j\neq i$ only if $V(G_j)\subset V(G_i)$ and
  hence $j<i$. Let $A\in\Real^{m\times m}$ be the matrix with entries
  $A_{ij}:=\sepi((G_i,T_i),G_j)$. Then $A$ is a lower triangular matrix
  with diagonal entries $A_{ii}=1$ for all $i$. This implies that $A$
  is invertible. 

  Let $\vec c=(c_1,\ldots,c_m)^T$ be the vector with entries
  $c_i:=\hom(G_i,H)$, and let $\vec b=(b_1,\ldots,b_m)^T$ be the
  vector with entries $b_i:=\pihom((G_i,T_i),H)$. By
  Corollary~\ref{cor:1}, for every $i$ we have
  \begin{align*}
    c_i&=\hom((G_i,T_i),H)\\
       &=\sum_{G'\colres G_i}\sepi((G_i,T_i),G')\cdot\pihom((G',T_i[V(G')]),H)\\
    &=\sum_{j=1}^mA_{ij}b_j.
  \end{align*}
  Thus $\vec c=A\vec b$, and since $A$ is invertible, $\vec
  b=A^{-1}\vec c$.

  Now let  $\vec c'=(c'_1,\ldots,c'_m)^T$ be the vector with entries
  $c'_i:=\hom(G_i,H')$, and let $\vec b'=(b'_1,\ldots,b'_m)^T$ be the
  vector with entries $b'_i:=\pihom((G_i,T_i),H')$. Then $\vec b'=A^{-1}\vec c'$. 

  By the assumption
  of the lemma, we have $\vec c=\vec c'$. Thus $\vec b=\vec b'$. In
  particular, 
  \[
    \pihom((G,T),H)=b_m=b_m'=\pihom((G,T),H').
    \qedhere
  \]
\end{proof}


Let us now move on to past-preserving homomorphisms.

\begin{lemma}
  Let $D\in\CD_k$, and let $h:D\to H$ be a past-injective
  homomorphism from $D$ to a graph $H$. Then there is a unique graph
  $G\supseteq F^D$ with $V(G)=V(D)$ such that $T^D$ is an elimination
  tree of\/ $G$ and $h$ is a past-preserving homomorphism from
  $(G,T^D)$ to $H$.
\end{lemma}

\begin{proof}
  Suppose that $D=(F,T)$.
  We let $G$ be the graph with $V(G):=V(F)$,
  \[
  E(G):=\Big\{uv\in\binom{V(F)}{2}\Bigmid u\dagle^T v\text{ and }h(u)h(v)\in E(H)\Big\},
  \]
  and $\col^G:=\col^F$.
  Then $G\supseteq F$, because $h$ is a homomorphism, and $T$ is an elimination tree of $G$, because $uv\in E(G)$ implies
  $u\dagle^T v$ or $v\dagle^T u$. Moreover, $h$ is past-preserving,
  because it is past-injective and for $u\dagle^T v$ we have $uv\in
  E(G)\iff h(u)h(v)\in E(H)$.

  It remains to prove the uniqueness. Let $G'\supseteq F$ with
  $V(G')=V(F)$ such that $T$ is an elimination tree of $G'$ and $h$ is a past-preserving homomorphism from
  $(G',T)$ to $H$. Then $\col^{G'}=\col^F=\col^G$, because $h$ is a
  homomorphism from $G'$ to $H$. Moreover, for all $uv\in E(G')$, either $u\dagle^T v$ or
  $v\dagle^T u$, because $T$ is an elimination tree of $G'$, and $uv\in
  E(G')\iff h(u)h(v)\in E(H)$, because $h$ is past-preserving. Thus
  $E(G')=E(G)$ and therefore $G=G'$.
\end{proof}

\begin{corollary}\label{cor:pp1}
    Let $D\in\CD_k$, and let $H$ be a graph. Then
  \[
    \pihom(D,H)=\hspace{-5mm}\sum_{\substack{G\supseteq F^D\textup{ such that
    }V(G)=V(D)\\\textup{and }(G,T^D)\in\CD_k}}\hspace{-10mm}\pphom((G,T^D),H).
  \]
\end{corollary}

\begin{corollary}\label{cor:pp2}
    Let $D\in\CD_k$, and let $H,H'$ be graphs such that for all
    $G\supseteq F^D$ with $V(G)=V(D)$ and $(G,T^D)\in\CD_k$,
  \[
    \pphom((G,T^D),H)=\pphom((G,T^D),H').
  \]
  Then 
  \[
    \pihom(D,H)=\pihom(D,H').
  \]
\end{corollary}

\begin{lemma}\label{lem:pp4}
  Let $D\in\CD_k$, and let $H,H'$ be graphs. Suppose that for all
    $G\supseteq F^D$ with $V(G)=V(D)$ and $(G,T^D)\in\CD_k$ we have
  \[
    \pihom((G,T^D),H)=\pihom((G,T^D),H').
  \]
  Then
  \[
    \pphom(D,H)=\pphom(D,H').
  \]
\end{lemma}

\begin{proof}
  Let $D=(F,T)$.
  Let $\CG$ be the set of all $G\supseteq F$ such that $V(G)=V(F)$ and
  $T$ is an elimination tree of $G$. In particular,
  $F\in\CG$. Let $G_1,\ldots,G_m$ be an enumeration of
  $\CG$ such that $G_1=F$ and $|E(G_i)|\le |E(G_j)|$ for $i\le
  j$. Let $A\in\Real^{m\times m}$ be the matrix with entries
  $A_{ij}:=1$ if $G_i\subseteq G_j$ and $A_{ij}=0$ otherwise. Then $A$ is an upper triangular matrix
  with diagonal entries $A_{ii}=1$ for all $i$. This implies that $A$
  is invertible. 

  Let $\vec c=(c_1,\ldots,c_m)^T$ be the vector with entries
  $c_i:=\pihom((G_i,T),H)$, and let $\vec b=(b_1,\ldots,b_m)^T$ be the
  vector with entries $b_i:=\pphom((G_i,T),H)$. By
  Corollary~\ref{cor:pp1}, we have $\vec c=A\vec b$, and since $A$ is invertible, $\vec
  b=A^{-1}\vec c$.

  Now let  Let $\vec c'=(c'_1,\ldots,c'_m)^T$ be the vector with entries
  $c'_i:=\pihom((G_i,T),H')$, and let $\vec b'=(b'_1,\ldots,b'_m)^T$ be the
  vector with entries $b'_i:=\pphom((G_i,T),H')$. Then $\vec b'=A^{-1}\vec c'$. 

  By the assumption
  of the lemma, we have $\vec c=\vec c'$. Thus $\vec b=\vec b'$. In
  particular, 
  \[
    \pihom(D,H)=b_1=b_1'=\pihom(D,H').
    \qedhere
  \]
\end{proof}

\begin{theorem}\label{theo:pp}
  For all $k\ge 0$ and all graphs $G,G'$, the following are equivalent.
  \begin{eroman}
     \item For all $F\in\TD_k$,
    \[
      \hom(F,G)=\hom(F,G').
    \]
  \item For all $D\in\CD_k$, 
    \[
      \pihom\big(D,G\big)=\pihom\big(D,G'\big).
    \]
  \item For all $D\in\CD_k$,
    \[
      \pphom\big(D,G\big)=\pphom\big(D,G'\big).
    \]
  \end{eroman}
\end{theorem}

\begin{proof}
  The implication (i)$\implies$(ii) follows from
  Lemma~\ref{lem:3}.

  As for all connected $F\in\cTD_k$ there is an elimination tree $T$ such that $(F,T)\in\CD_k$, it follows from
  Corollary~\ref{cor:2} that (ii) implies $ \hom(F,G)=\hom(F,G')$ for
  all $F\in\cTD_k$. But we have observed earlier that this implies $ \hom(F,G)=\hom(F,G')$ for
  all $F\in\TD_k$.
  
  The equivalence between (ii) and (iii) follows from
  Lem\-ma~\ref{lem:pp4} and Corollary~\ref{cor:pp2}.
\end{proof}

\section{Playing the Game}

In this section, we will connect the numbers of past-preserv\-ing
homomorphisms to the bijective pebble game and use this to prove our
main theorem. We start with a technical lemma that we need for our
interpolation arguments later.\footnote{We are convinced that this
  lemma is known to many other researchers, but lacking a
  reference, we decided to include a proof.}

\begin{lemma}\label{lem:4}
 Let $\vec a_1,\ldots,\vec a_\ell\in\Nat^m$, where $\vec
  a_i=(a_{i1},\ldots,a_{im})$, be mutually
  distinct vectors with positive entries. For every $i\in[\ell]$ and every $\vec
  d=(d_1,\ldots,d_m)\in\Nat^m$, let $\vec a_i^{(\vec d)}:=\prod_{j=1}^ma_{ij}^{d_j}$.

  Then there is a $\vec d=(d_1,\ldots,d_m)\in\Nat^m$ such that $1\le d_i\le\ell^2$ for all $i\in[m]$ and
  $\vec a_1^{(\vec d)},\ldots,\vec a_\ell^{(\vec d)}$ are mutually distinct. 
\end{lemma}

\begin{proof}
  The proof is by induction on $m$. The case $m=1$ is trivial, we can
  simply choose $d_1=1$. For the inductive step $m-1$ to $m$, let
  $\vec a_i':=(a_{i1},\ldots,a_{i(m-1)})$. Let $P_1,\ldots,P_{\ell'}$
  be the partition of $[\ell]$ such that $\vec a_i'=\vec a_j'$ if and
  only if $i,j\in P_p$ for some $p\in[\ell']$. For $p\in[\ell']$, let
  $\vec b_p:=\vec a_i'$ for $i\in P_p$.  By the induction hypothesis,
  there is a vector $\vec d':=(d_1,\ldots,d_{m-1})$ such that
  $1\le d_i\le(\ell')^2\le\ell^2$ for all $i\in[m-1]$ and the numbers
  $\vec b_1^{(\vec d')},\ldots,\vec b_\ell^{(\vec d')}$ are mutually
  distinct. In the following, we keep $d_1,\ldots,d_{m-1}$ fixed and
  try to find a $d_m$ such that $\vec d=(d_1,\ldots,d_m)$ satisfies
  the assertion of the lemma.

    Observe that for $d_m\ge 1$ and $i,j\in[\ell]$, if $i,j\in
    P_p$ for some $p$ then $\vec a_i'=\vec a_j'$ and $a_{im}\neq
    a_{jm}$ and thus
    \[
      \vec a_i^{(\vec d)}=\vec b_p^{(\vec d')}a_{im}^{d_m}\neq \vec b_p^{(\vec
        d')}a_{jm}^{dm}=\vec a_{j}^{(\vec d)}.
    \]
    For $i\in P_p,j\in P_q$ with $p\neq q$ we have 
    \[
      \vec a_i^{(\vec d)}=\vec a_{j}^{(\vec d)}\iff\frac{\vec
        b_p^{(\vec d')}}{\vec b_q^{(\vec
          d')}}=\frac{a_{jm}^{d_m}}{a_{im}^{d_m}}.
      \]
      If $a_{im}= a_{jm}$, there is no such $d_m$, because $\frac{\vec
        b_p^{(\vec d')}}{\vec b_q^{(\vec
          d')}}\neq 1$. If $a_{im}\neq
      a_{jm}$, there is at most one such $d_m$. Overall, for all
      distinct $i,j\in[\ell]$ there is at most one $d_m\ge 1$ such
      that $\vec a_i^{(\vec d)}=\vec a_{j}^{(\vec d)}$. Thus by the
      pigeonhole principle, there is a
      $d_m\le\binom{\ell}{2}+1\le\ell^2$ such that $\vec a_i^{(\vec
        d)}\neq\vec a_{j}^{(\vec d)}$ for all distinct $i,j$.
  \end{proof}

  Let $D_1,\ldots,D_m\in\CD$ such that the roots $r_i:=r^{D_i}$ all
  have the same colour, that is, $\col^{D_i}(r_i)=\col^{D_j}(r_j)=:c$
  for all $i,j\in[m]$. We say that $D_1,\ldots,D_m$ are
  \alert{compatible}. The \alert{rooted sum} of $D_1,\ldots,D_m$ is
  the pair $D=(F,T)$ where $F$ is the graph obtained from the disjoint
  union of the $F^{D_i}$ by identifying the roots $r_1,\ldots,r_m$ and
  $T$ is the tree obtained from the disjoint union of the trees
  $T^{D_1},\ldots,T^{D_m}$ by identifying their roots.  We write
  $ \alert{D=\bigoplus_{i=1}^m D_i} $ to express that $D$ is the
  rooted sum of the $D_i$.  For $d\ge 1$, we write
  \alert{$D=d\odot D'$} to express that $D$ is the rooted sum of $d$
  disjoint copies of $D'$. We combine these notations, writing
\[
\alert{D=\bigoplus_{i=1}^md_i\odot D_i}
\]
to express that $D$
is the rooted sum of $d_i$ disjoint copies of $D_i$ for each
$i\in[m]$.
For 
every set $\CF\subseteq\CD$ we let \alert{$\CF^{\oplus}$} denote the set of 
all rooted sums of elements of $\CF$. 

Recall that for a $D\in\CD$, a graph $G$, and vertices $u\in V(D),v\in
V(G)$, by $\pphom\big(D,G;u\mapsto v\big)$ we denote the number of
past-preserving homomorphisms $h:D\to G$ with $h(u)=v$.
 Observe
that if $D=\bigoplus_{i=1}^m D_i$ for $D_i\in\CD_k$,
then $D\in\CD_k$ and for all graphs $G$ and vertices $v\in V(G)$ we have
\begin{equation}
  \label{eq:2}
  \pphom(D,G;r^D\mapsto v)=\prod_{i=1}^m \pphom(D_i,G;r^{D_i}\mapsto v). 
\end{equation}

\begin{lemma}\label{lem:5}
  Let $\CF\subseteq\CD$. Let $G,G'$ be graphs
  such that $|G|=|G'|$ and 
  \begin{equation}
    \label{eq:3}
    \pphom\big(D,G\big)=\pphom\big(D,G'\big)
  \end{equation}
  for all $D\in\CF^\oplus$. Then there is a bijection $f:V(G)\to V(G')$ such that 
  \[
    \pphom\big(D,G;r^D\mapsto v\big)=\pphom\big(D,G'; r^D\mapsto f(v)\big)
  \]
  for all $D\in\CF$.
\end{lemma}

\begin{proof}
  Let $n:=|G|=|G'|$.
  Without loss of generality, we assume that $V(G)\cap V(G')=\emptyset$.
    We define an equivalence relation $\sim$ on $V(G)\cup V(G')$ as follows:
  for $X,Y\in\{G,G'\}$ and $x\in V(X)$, $y\in V(Y)$, we let $x\sim y$
  if and only if for all 
  $D\in\CF$,
  \begin{equation}
    \label{eq:20}
    \pphom\big(D,X;r^D\mapsto x\big)=\pphom\big(D,Y; r^D\mapsto y\big).
  \end{equation}
  Let $K_1,\ldots, K_\ell$ be the $\sim$-equivalence
    classes.  For every $i\in [\ell]$, let $p_i:=|K_i\cap V(G)|$ and
    $p'_i:=|K_i\cap V(G')|$, and let $\vec p=(p_1,\ldots,p_\ell)$ and
    $\vec p'=(p'_1,\ldots,p'_\ell)$.

  The assertion of the lemma is an immediate consequence of the
    following claim.

    \begin{claim}
      \[
        \vec p=\vec p'.
      \]

      \proof Without loss of generality, we assume that $\CF$ is
      finite. If it is not, for all distinct $i,j\in[\ell]$ we pick a
      $D_{ij}\in\CF$ such that for $x\in K_i,y\in K_j$ and
      $X,Y\in\{G,G'\}$ with $x\in V(X)$, $y\in V(Y)$ it holds that
      \[
      \pphom\big(D_{ij},X;r^{D_{ij}}\mapsto x\big)\neq\pphom\big(D_{ij},Y; r^{D_{ij}}\mapsto y\big),
      \]
      and we restrict our attention to the finite class of all these
      $D_{ij}$ without changing the equivalence relation $\sim$. 

      Say,
      $\CF=\{D_1,\ldots,D_m\}$, and for every $i\in[m]$,
      let $r_i:=r^{D_i}$. Then for all $X,Y\in\{G,G'\}$ and $x\in
      V(X),y\in V(Y)$ we have $x\sim y$ if and only if
      \[
        \forall j\in[m]:  \pphom(D_j,X;r_j\mapsto x)=\pphom(D_j,Y; r_j\mapsto y).
    \]
    For $i\in[\ell]$, let $a_{ij}:=\pphom(D_j,X;r_j\mapsto
    x)$ for all $X\in\{G,G'\}$, $x\in K_i\cap V(X)$. Let $\vec a_i:=(a_{i1},\ldots,a_{im})$ and $\vec
    a^j=(a_{1j},\ldots,a_{\ell j})^T$. Thus the $\vec a_i$ are the
    rows and the $\vec a^j$ the columns of the $(\ell\times m)$-matrix
    with entries $a_{ij}$. Observe that the rows $\vec a_i$ are
    mutually distinct.
    
    For every $j\in[m]$ we have
    \begin{align*}
      \pphom(D_j,G)&=\sum_{v\in V(G)}\pphom(D_j,G;r_j\mapsto
                          v)\\
                  &=\sum_{i=1}^\ell p_i a_{ij}\\
                  &=\angles{\vec p,\vec a^j}.
                    \intertext{and similarly}
                    \pphom(D_j,G')&=\angles{\vec p',\vec a^j}.
    \end{align*}
    Here $\angles{\vec x,\vec y}=\sum_{i}x_iy_i$ denotes the standard
    inner product of vectors $\vec x,\vec y$.

    Let us
     call a set 
     $J\subseteq [m]$ \alert{compatible} if all $D_j$ for $j\in J$ are
     compatible (that is, their roots have the same colour). The \alert{support} of a vector $\vec d=(d_1,\ldots,d_m)$
     is the set $\supp(\vec d):=\{i\in[m]\mid d_i\neq 0\}$, and we
     call $\vec d$ \alert{compatible} if its support is compatible. For
     every
     compatible nonzero vector
     $\vec d=(d_1,\ldots,d_m)\in\Nat^m$ we let
     \[
     D^{(\vec d)}:=\bigoplus_{j\in\supp(\vec d)}d_j\odot D_j.
     \]
     Then $D^{(\vec d)}\in\CF^{\oplus}$. We denote the root of
     $D^{(\vec d)}$ by $r^{(\vec d)}$.
    By \eqref{eq:2}, for every $X\in\{G,G'\}$ and $x\in V(X)$ we have
     \[
       \pphom\big(D^{(\vec d)},X;r^{(\vec d)}\mapsto
       x\big)=\prod_{j=1}^m\pphom(D_j,X;r_j\mapsto
       x)^{d_j}.
     \]
     Thus for every $i\in[\ell]$, $X\in\{G,H\}$, and  $x\in K_i\cap V(X)$ we have
     \begin{equation}
       \label{eq:4}
       \pphom\big(D^{(\vec d)},X;r^{(\vec d)}\mapsto x\big)=\prod_{j=1}^m
     a_{ij}^{d_j}=:\vec a_i^{(\vec d)}.
   \end{equation}
     Let
     $\vec a^{(\vec d)}=(\vec a_{1}^{(d)},\ldots,\vec a_{\ell }^{(d)})^T$. Note that
     with this notation, $\vec a^j=\vec a^{(\vec e_j)}$, where $\vec
     e_j$ denotes the $j$th unit vector.
Then
     \begin{align*}
       \pphom(D^{(\vec d)},G)&=\sum_{v\in V(G)}\pphom(D^{(\vec d)},G;r^{(d)}\mapsto
                          v)\\
                  &=\sum_{i=1}^\ell p_i \vec a_{i}^{(d)}\\
                  &=\angles{\vec p,\vec a^{(d)}}.
                    \intertext{and similarly}
       \pphom(D^{(\vec d)},G')&=\angles{\vec p',\vec a^{(\vec d)}}.
     \end{align*}
     For all $j$, let
     $b^{(\vec d)}:=\pphom(D^{(\vec d)},G)$. Since $D^{(\vec d)}\in\CF^{\oplus}$, by \eqref{eq:3}, we have
     $b^{(\vec d)}=\pphom(D^{(\vec d)},G')$. Then
     \begin{equation}
       \label{eq:5}
              b^{(\vec d)}=\angles{\vec p,\vec a^{(\vec d)}}=\angles{\vec p',\vec a^{(\vec d)}}. 
     \end{equation}
        Since $0^0=1$ and $0^d=0$ for $d\ge 1$, for every $i\in[\ell]$
        and every $\vec d\in\Nat^m$ we have
     \begin{equation}
       \label{eq:6}
       \vec a_i^{(\vec d)}\neq0\iff \supp(\vec d)\subseteq\supp(\vec a_i).
     \end{equation}
     Suppose for contradiction that $\vec p\neq \vec p'$. Choose
     $i_0\in[\ell]$ such that $p_{i_0}\neq
     p'_{i_0}$ and, subject to this condition, $S:=\supp(\vec a_{i_0})$ is
     inclusionwise maximal. 

     Suppose first that $S=\emptyset$. Then $\vec a_{i_0}=\vec 0$, and
     as the $\vec a_i$ are mutually distinct, $\vec a_i\neq\vec 0$ and
     therefore $\supp(\vec a_i)\neq\emptyset$ for $i\neq i_0$. By the maximality of
     $S$, this implies $p_{i}=p'_{i}$ for all $i\neq i_0$. Hence
     $p_{i_0}=n-\sum_{i\neq i_0}p_{i}=n-\sum_{i\neq i_0}p'_{i}=p'_{i_0}$, which is
     a contradiction. It follows that $S\neq\emptyset$.

     Let $I=\{i\in[\ell]\mid\supp(\vec a_i)=S\}$. For every $i\in I$, let
     $\hat{\vec a}_i:=(a_{ij}\mid j\in S)$. Then the vectors
     $\hat{\vec a}_i$ have only positive entries, and they are
     mutually distinct, because the $\vec a_i$ are mutually
     distinct. By Lemma~\ref{lem:4}, there is a vector $\hat{\vec
       d}=(\hat d_j\mid j\in S)$ such that the numbers $\hat{\vec a}_i^{(\hat{\vec d})}$
     for $i\in I$ are mutually distinct. Let $\vec
     d=(d_1,\ldots,d_m)$ with $d_j=\hat d_j$ for $j\in S$ and $d_j=0$
     otherwise. Then the numbers $\vec a_i^{(\vec d)}$
     for $i\in I$ are mutually distinct.

     Observe that for every $j\in \Nat$ we
     have $\vec a_i^{(j\vec d)}=\big(\vec a_i^{(\vec d)}\big)^j$, where
     $j\vec d=(jd_1,\ldots,jd_m)$. Let $A$ be the
     $|I|\times|I|$-matrix with entries
     $A_{ij}:=\vec a_i^{(j\vec d)}$ (for convenience, we take row
     indices from the set $I$ and column indices from
     $\{1,\ldots,|I|\}$). $A$ is a Vandermonde matrix and thus
     invertible. 

     Let $\vec p_I:=(p_i\mid i\in I)$ and
     $\vec p'_I:=(p'_i\mid i\in I)$ be the restrictions of $\vec p$ and
     $\vec p'$ to $I$. For every $j$, the $j$th entry of
     $\vec p_I\cdot A$ is
     \begin{align*}
       \sum_{i\in I}p_i\vec a_i^{(j\vec d)}&=\sum_{i=1}^\ell p_i
                                             \vec a_i^{(j\vec
                                             d)}-\hspace{-2mm}\sum_{\substack{i\in[\ell]\\
                                               S\not\subseteq
                                               \supp(\vec a_i)}} p_i
                                             \vec a_i^{(j\vec
                                             d)}-\hspace{-2mm}\sum_{\substack{i\in[\ell]\\
                                               S\subset
                                               \supp(\vec a_i)}} p_i
                                             \vec a_i^{(j\vec
                                             d)}\\
&=\angles{\vec p,\vec a^{(j\vec d)}}-\sum_{\substack{i\in[\ell]\\
                                               S\subset
                                               \supp(\vec a_i)}} p_i
                                             \vec a_i^{(j\vec
                                             d)},
     \end{align*}
     because by \eqref{eq:6} we have $a_i^{(j\vec d)}=0$ if
     $S=\supp(j\vec d)\not\subseteq \supp(\vec a_i)$. Similarly,
     the $j$th
     entry of $\vec p'_I\cdot A$ is
     \[
       \sum_{i\in I}p'_i\vec a_i^{(j\vec d)}=\angles{\vec p',\vec a^{(j\vec d)}}-\sum_{\substack{i\in[\ell]\\
           S\subset
           \supp(\vec a_i)}} p'_i
       a_i^{(j\vec
         d)}.
     \]
     By the maximality of $S$, for all $i$ with $S\subset\supp(\vec
     a_i)$ we have  $p_i=p'_i$. Thus 
     \[
       \sum_{\substack{i\in[\ell]\\
           S\subset
           \supp(\vec a_i)}} p_i
       \vec a_i^{(j\vec
         d)}=\sum_{\substack{i\in[\ell]\\
           S\subset
           \supp(\vec a_i)}} p'_i
       \vec a_i^{(j\vec
         d)},
     \]
     and therefore, by \eqref{eq:5},
     \[
       \sum_{i\in I}p_i\vec a_i^{(j\vec d)}=\sum_{i\in I}p'_i\vec
       a_i^{(j\vec d)},
     \]
     Since this holds for all $j$, we have $\vec p_I\cdot A=\vec
     p'_I\cdot A$. As
     $A$ is invertible, it follows that $\vec p_I=\vec p'_I$ and, in
     particular, $p_{i_0}=p'_{i_0}$. This is a contradiction.
     \qedhere
    \end{claim}
\end{proof}

\begin{remark}\label{rem:numclasses}
  The proof of the lemma actually shows that \eqref{eq:3} only needs
  to be satisfied for rooted sums of at most $m\ell^2$ graphs from $\CF$,
  where $m:=|\CF|$, and $\ell$ is the number of equivalence classes of
  the relation $\sim$. (Recall the definition of $\sim$ from
  \eqref{eq:20}.) Note that we always have $\ell\le n=|G|$.
\end{remark}

\begin{corollary}\label{cor:21}
  Let $G,G'$ be graphs such that for all $D\in\CD_k$,
  \begin{equation}
    \label{eq:7}
    \pphom\big(D,G\big)=\pphom\big(D,G'\big). 
  \end{equation}
  Then there is a bijection $f:V(G)\to V(G')$ such that for all $D\in\CD_k$,
  \[
    \pphom\big(D,G;r^D\mapsto v\big)=\pphom\big(D,G'; r^D\mapsto f(v)\big).
  \]
\end{corollary}

\begin{proof}
  This follows immediately from Lemma~\ref{lem:5}, noting that
  $\CD_k^\oplus=\CD_k$ and that \eqref{eq:7} for all $D\in\CD_k$ implies that $|G|=|H|$.
\end{proof}




\subsection{Proof of the Main Theorem}
\label{sec:main-proof}

For the inductive proof, the following construction is useful.
Let $G$ be a graph and $v\in V(G)$. We let \alert{$G\wr v$} be the graph with
vertex set $V(G\wr v):=V(G)\setminus\{v\}$, edge set $E(G\wr
v):=\{ww'\in E(G)\mid w,w'\in V(G)\setminus \{v\}\}$, and colouring
defined by
\[
\col^{G\wr v}(w):=
\begin{cases}
  \big(\col^G(w),1\big)&\text{if }vw\in E(G),\\
  \big(\col^G(w),0\big)&\text{if }vw\not\in E(G).
\end{cases}
\]

\begin{lemma}\label{lem:wr}
  Let $k\ge 0$, and let $G,G'$ be graphs and $v\in
  V(G),v'\in V(G')$ such that 
  $|G|=|G'|\ge 2$ and $\col^G(v)=\col^{G'}(v')$. Then the following are equivalent.
  \begin{eroman}
  \item Duplicator has a winning strategy for the $k$-round bijective
    pebble game on $G,G'$ with initial position $(v,v')$.
  \item Duplicator has a winning strategy for the $k$-round bijective
    pebble game on $G\wr v$, $G'\wr v'$.
  \end{eroman}
\end{lemma}

\begin{proof}
  Straightforward.
\end{proof}

The next lemma is the last significant step of the proof of our
main theorem. After that, we only need to pull things together to
complete the proof.

\begin{lemma}\label{lem:6}
  Let $k\ge 1$, and let $G,G'$ be graphs of the same order. Then
  the following are equivalent.
  \begin{eroman}
    \item For all $D\in\CD_k$, 
      \[
    \pphom(D,G)=\pphom(D,G').
    \]
  \item Duplicator has a winning strategy for the $k$-round bijective
  pebble game on $G,G'$.
  \end{eroman}
\end{lemma}

\begin{proof}
  We first prove (i)$\implies$(ii). The proof is by induction on $k$. 

  For the base case $k=1$, suppose that $\pphom(D,G)=\pphom(D,G')$
  for all $D\in\CD_1$. By Corollary~\ref{cor:21}, there is a bijection
  $f:V(G)\to V(G')$ such that for all $D\in\CD_1$ and $v\in V(G)$,
  \[
 \pphom\big(D,G;r^D\mapsto v\big)=\pphom\big(D,G'; r^D\mapsto f(v)\big).
 \]
 This implies $\col^G(v)=\col^{G'}(f(v))$. Duplicator picks $f$ in the
 first (and only) round of the game and wins.

 For the inductive step $k\to k+1$, let $G,G'$ be graphs of the same
 order such that $\pphom(D,G)=\pphom(D,G')$
  for all $D\in\CD_{k+1}$. If
  $|G|=|G'|=1$, then $\pphom(D,G)=\pphom(D,G')$ for all $D\in\CD_1$ implies that the
  graphs are isomorphic (their unique vertices have the same colour). Thus
  we may further assume that $|G|=|G'|\ge 2$. 

By  Corollary~\ref{cor:21}, there is a bijection
  $f:V(G)\to V(G')$ such that for all $D\in\CD_{k+1}$ and $v\in V(G)$,
  \begin{equation}
    \label{eq:8}
    \pphom\big(D,G;r^D\mapsto v\big)=\pphom\big(D,G'; r^D\mapsto f(v)\big).
  \end{equation}
 In the first round of the game, Duplicator picks this bijection $f$. Say,
 Spoiler picks $v\in V(G)$. Let $v':=f(v)$. Note that \eqref{eq:8}
 implies $\col^G(v)=\col^{G'}(v')$. Let $H:=G\wr v$ and $H':=G'\wr
 v'$. We need to prove that Duplicator has a winning strategy for the
 remaining $k$-round bijective pebble game on $G,G'$ with initial
 position $(v,v')$. By Lemma~\ref{lem:wr}, it suffices to prove that Duplicator has
 a winning strategy for the $k$-round bijective pebble game on $H,H'$.
 This
 follows immediately from the induction hypothesis and the following
 claim.

 \begin{claim}
   For all $D\in\CD_{k}$, 
   \[
     \pphom(D,H)=\pphom(D,H').
   \]

   \proof
   Let $D=(F,T)\in\CD_{k}$.
   We may assume that all vertices $u\in V(D)$ have a colour of the
   form $(c,i)$ where $i\in\{0,1\}$. Otherwise,
   $\pphom(D,H)=\pphom(D,H')=0$. 

   We define a graph $F^+$ and a tree $T^+$ as follows. We take a
   fresh vertex $r^+$ and let $V(F^+):=V(F)\cup\{r^+\}$, 
   \begin{align*}
     E(F^+):=&E(F)\,\cup\\
     &\{r^+u\mid u\in V(F)\text{ with
     }\col^{F}(u)=(c,1)\text{ for some }c\},
   \end{align*}
   and 
   $\col^{F^+}(r^+):=\col^G(v)=\col^{G'}(v')$ and $\col^{F^+}(u):=c$ for all
   $u\in V(F)$ with $\col^F(u)=(c,i)$ for some $i\in\{0,1\}$. We let
   $V(T^+):=V(T)\cup\{r^+\}$ and 
   \[
     \dagle^{T^+}:=\dagle^T\cup\{(r^+,u)\mid u\in V(T)\}.
   \]
   Then $r^+$ is the root of $T^+$. Then $T^+$ is an
   elimination tree of $F^+$ of height $k+1$. Hence $D^+:=(F^+,T^+)\in\CD_{k+1}$.

      Observe that
   that there is a one-to-one correspondence between the past-preserving
   homomorphisms from $D^+$ to $G$ mapping $r^+$ to $v$ and the
   past-preserving homomorphisms from $D$ to $H$.\footnote{Note that here we
   need the homomorphisms to be past preserving. The proof would break
   down if we worked with arbitrary homomorphisms, because a
   homomorphism  from $D^+$ to $G$ mapping $r^+$ to $v$ could also
   map vertices in $V(D^+)\setminus\{r^+\}=V(D)$ to $v$, and such a
   homomorphism would not correspond to a homomorphism from $D$ to
   $H$.}
 Similarly, there is
   a one-to-one correspondence between the past-preserving
   homomorphisms from $D^+$ to $G'$ mapping $r^+$ to $v'$ and the
   past-preserving homomorphisms from $D$ to $H'$.
   Thus 
   \begin{align*}
     \pphom(D,H)&=\pphom(D^+,G;r^+\mapsto
                         v),\\
     \pphom(D,H')&=\pphom(D^+,G';r^+\mapsto v').
   \end{align*}
   By \eqref{eq:8} applied to $D^+$, we have
   \[
     \pphom(D^+,G;r^+\mapsto v)=\pphom(D^+,G';r^+\mapsto v').
   \]
   Thus
   $\pphom(D,H)=\pphom(D,H')$.
   \uend
 \end{claim}

 The proof of the converse direction (ii)$\implies$(i) is also by
 induction on $k$.

 For the base case $k=1$, assume that Duplicator has a winning strategy for
 the $1$-move bijective pebble game on $G,G'$. Then there is a
 bijection $f:V(G)\to V(G')$ such that $\col^G(v)=\col^{G'}(f(v))$ for
 all $v\in V(G)$,
 which implies that for each colour $c$ the two graphs have the same
 numbers of vertices of colour $c$. This implies that
 $\pphom(D,G)=\pphom(D,G')$ for all $D\in\CD_1$.

 For the inductive step $k\to k+1$, assume that Duplicator has a
 winning strategy for the $(k+1)$-round bijective pebble game on
 $G,G'$. Without loss of generality we may assume that $|G|=|G'|\ge
 2$. Then, by Lemma~\ref{lem:wr}, there is a bijection $f:V(G)\to
 V(G')$ such that for each $v\in V(G)$, $\col^G(v)=\col^{G'}(f(v))$
 and Duplicator has a winning
 strategy for the $k$-round bijective pebble game on $G\wr v,G'\wr
 f(v)$. By the induction hypothesis, this implies
 \begin{equation}
   \label{eq:9}
   \pphom(D,G\wr v)=\pphom(D,G'\wr f(v)) 
 \end{equation}
 for all $D\in\CD_{k}$. 

 Now let $\hat D\in\CD_{k+1}$ and $\hat r:=r^{\hat D}$. By deleting
 $\hat r$ from
 $\hat D$, we obtain a family $D_1,\ldots,D_m\in \CD_{k}$.
 For every $i\in[m]$, let $\hat D_i\in\CD_k$ be obtained from $D_i$ by
 recolouring the vertices as follows: for $u\in V(D_i)$, let
 \[
   \col^{\hat D_i}(u):=
   \begin{cases}
     (\col^{D_i}(u),1)&\text{if }\hat r u\in E(\hat D),\\
     (\col^{D_i}(u),0)&\text{otherwise}.
   \end{cases}
 \]
 The
 crucial observation is that for each $v\in V(G)$ with
 $\col^G(v)=\col^D(\hat r)$ we have
 \begin{align*}
 \pphom(\hat D,G;\hat r\mapsto v)&=\prod_{i=1}^m\pphom(\hat D_i,G\wr v)\\
 \intertext{and similarly}
 \pphom(\hat D,G';\hat r\mapsto f(v))&=\prod_{i=1}^m\pphom(\hat D_i,G'\wr f(v)).
 \end{align*}
 By \eqref{eq:9}, this implies
 \[
   \pphom(\hat D,G;\hat r\mapsto
   v)=\pphom(\hat D,G';\hat r\mapsto f(v)).
 \]
 Thus
 \begin{align*}
 \pphom(\hat D,G)
 &=\sum_{\substack{v\in V(G)\\\col^G(v)=\col^{\hat D}(\hat r)}}\pphom(\hat
 D,G;\hat r\mapsto v)\\
 &=\sum_{\substack{v\in V(G)\\\col^G(v)
     =\col^{\hat D}(\hat r)}}\pphom(\hat D,G';\hat r\mapsto f(v))\\
&=\sum_{\substack{v'\in V(G')\\\col^{G'}(v')=\col^{\hat D}(\hat r)}}\pphom(\hat
 D,G';\hat r\mapsto v')\\
   &=\pphom(\hat D,G').
\end{align*}
\end{proof}

\begin{proof}[Proof of Theorem~\ref{theo:main}]
  The theorem follows from the previous lemma combined with
  Lemma~\ref{lem:bpgame} (stating that winning strategies for
  Duplicator in the bijective pebble game establish equivalence in the
  logic) and Theorem~\ref{theo:pp} (stating the equivalence between
  homomorphism counts and past-preserving homomorphism counts),
  observing that for all $k\ge 1$, graphs $G,G'$ of distinct orders
  are neither homomor\-phism-indistinguishable over the class $\TD_k$
  nor $\LC_{k}$-equi\-valent.
\end{proof}
 
\section{Discussion}

It is a consequence of our main theorem that every sentence $\phi$ of the
logic $\LC$ and other counting logics such as Kuske and Schweikardt's
\cite{kusschwe17} $\textsf{FOCN}(\mathbf P)$ is equivalent to an
infinitary Boolean combination of expressions of
\alert{$\eta_{F,m}$} stating that ``there are exactly $m$ homomorphism from $F$ into
the current graph'', where $F\in\TD_k$ for the quantifier rank $k$ of
$\phi$. Indeed, it follows immediately from Theorem~\ref{theo:main} that
every sentence $\phi\in\LC$ of quantifier rank $k$ is
equivalent to
\begin{equation}
  \label{eq:21}
    \bigvee_{\substack{G\text{ graph}\\\text{such that }G\models\phi}}
  \bigwedge_{F\in\cTD_k}\eta_{F,\hom(F,G)}.
\end{equation}
Observe that $\eta_{F,m}$
can be viewed as a sentence of the form $\exists^{=m}\vec x\alpha$,
where $\vec x$ is a tuple of $|F|$ variables and $\alpha$ is a
conjunction of atoms of the form $E(x_i,x_j)$ and
$\gamma(x_i)=c$. This gives us a normal form for $\LC$-sentences that
is local and achieves some form of quantifier elimination (or, maybe
more precisely, \emph{quantifier de-alternation}). But of course the
infinite disjunction and conjunctions are unpleasant. We can replace
the infinite conjunctions by a finite one, ranging over a finite set
$\CF(G)\subseteq\cTD_k$ that only depends on $G$. But there is no hope
of avoiding the infinite disjunction.

\subsection{Graphs of Bounded Degree}

For graphs $G$ of bounded degree, we can improve our main theorem.
We fix a set $\Gamma$ of colours and only consider graphs $G$ with $\rg(\gamma^G)\subseteq\Gamma$.

\begin{theorem}\label{theo:d}
  Let $k,d\ge 1$. Then there is a finite set
  $\CF_{k,d}\subseteq\TD_k$, computable from $k,d$, such that for all
  graphs $G,G'$ of maximum degree at most $d$ the following are equivalent.
  \begin{eroman}
  \item 
    $G$ and $G'$ are homomorphism-indistinguishable over $\CF_{k,d}$.
  \item
    $G$ and $G'$ satisfy the same $\LC_{k}$-sentences.
  \end{eroman}
\end{theorem}

\begin{proof}[Proof (sketch)]
  We only need to prove that the implication (i)$\implies$(ii) holds for a
  sufficiently large $\CF_{k,d}\subseteq\TD_k$.

  Recall that connected graphs in $\TD_k$ have radius at most
  $2^{k-1}-1$. Thus if two vertices $v,w$ in a graph $G$ have
  isomorphic neighbourhoods of radius $2^{k-1}-1$, then
  $\hom(D,G;r^D\mapsto v)=\hom(D,G;r^D\mapsto v')$ for all
  $D\in\CD_k$. Note that the equality holds even though the graph $F^D$ is not
  necessarily connected, because on both sides of the equality we have
  the same graph $G$, and every connected component of $D$ that does
  not contain the root $r^D$ contributes to both sides of the equation
  in the same way.

  In graphs of maximum degree at most $d$, the number of
  isomorphism types of neighbourhoods of radius $2^{k-1}-1$ is bounded
  in terms of $k$ and $d$. This means that there is only a bounded
  number of homomorphism counts $\hom(D,G;r^D\mapsto v)$. Recall
  Remark~\ref{rem:numclasses}. By what we have just observed, if both $G$
  and $G'$ are of maximum degree $d$, the number $\ell$ of equivalence
  classes is bounded in terms of $k,d$, and thus we only need to
  consider rooted sums of at most $f(k,d)$ graphs (for a suitable
  function $f$).

  If we plug this into the inductive proof of the main theorem, we see
  that we only need to consider homomorphism counts from $g(k,d)$
  graphs, for a suitable function $g$.
\end{proof}

Note that this stronger version of the theorem leads to a slight
improvement of the normal form \eqref{eq:21} for graphs of maximum
degree at most $d$: independently of the disjunct $G$, we can restrict
the conjunction to graphs $F$ from the finite set
$\CF_{k,d}$. 

\subsection{Equivalence in First-Order Logic}

We may wonder if in Theorem~\ref{theo:d} we really need the
dependence of the set $\CF_{k,d}$ on the maximum degree $d$. That is,
we may ask if for every $k$ there is a finite set
$\CF_{k}\subseteq\TD_k$ such that for all
graphs $G,G'$
if $G$ and $G'$ are homomorphism-indistinguishable over $\CF_{k}$ then
they are $\LC_k$-equivalent. It is easy to see that this cannot be
the case, essentially because the number of $\LC_k$-equivalence
classes is unbounded.

However, this is different for first-order logic $\FO$: for every $k$
there are only finitely many $\FO_k$-equivalence classes, where
$\FO_k$ denotes the fragment of $\FO$ consisting of all formulas of
quantifier rank at most $k$. Thus it may be tempting to conjecture the
following. Again, we fix a set $\Gamma$ of colours and only consider graphs $G$ with $\rg(\gamma^G)\subseteq\Gamma$.

\begin{conjecture}
  Let $k\ge 1$. Then there is a finite set $\CF_{k}\subseteq\TD_k$
  such that for all graphs $G,G'$, if $G$ and $G'$ are
  homo\-morphism-indistinguishable over $\CF_{k}$, then $G$ and $G'$
  satisfy the same $\FO_{k}$-sentences.
\end{conjecture}

It would be really nice if this conjecture was true. For example, it
would imply a parameterised version of Toda's theorem, settling a
long-standing open problem in parameterised complexity theory \cite{flugro04b}. \emph{Unfortunately,
  the conjecture is false already for $k=2$.}

\begin{figure}
  \centering
  \begin{tikzpicture}
    \foreach \x in {0,...,3} {
      \draw[thick] (3,1.5) -- (\x,0);
      \draw[fill=white] (\x,0) circle (1mm);
    }
    \foreach \x in {4,5,6} {
      \draw[thick] (3,1.5) -- (\x,0);
      \draw[fill=black] (\x,0) circle (1mm);
    }
    \draw[fill=grey] (3,1.5) circle (1mm);
  \end{tikzpicture}
  \caption{The star $S_{4,3}$}
  \label{fig:star}
\end{figure}
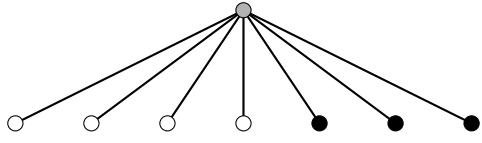

\begin{example}
  Let us assume that $\Gamma=\{\tikz{\draw (0,0) circle (1mm);},
  \tikz{\draw[fill=grey] (0,0) circle (1mm);}, \tikz{\draw[fill] (0,0)
    circle (1mm);}\}$. For all $k,\ell\in\Nat$, we let $S_{p,q}$ be
  the star with a centre $r$ and tips $s_1,\ldots,s_p,t_1,\ldots,t_q$
  such that $r$ is grey, the $s_i$ are white, and the $t_j$ are black
  (see Figure~\ref{fig:star}). Moreover, we let $W$ be the graph
  consisting of a single white vertex and $B$ the graph
  consisting of a single black vertex. Observe that
  $S_{0,0},W,B\in\TD_1$ and $S_{p,q}\in\TD_2$ for all $p,q\in\Nat$.

  Let $\CS$ be the class of all graphs that are finite disjoint unions
  of stars $S_{p,q}$ for $p,q\in\Nat$. For every graph $G\in\CS$ and
  all $p,q\in\Nat$, let $a_{p,q}(G)$ be the number of copies of
  $S_{p,q}$ in $G$. Observe that for all $i,j,p,q\in\Nat$ we have
  \[
    \hom(S_{i,j},S_{p,q})=p^iq^j.
  \]
  Thus for $G\in\CS$, 
  \begin{equation}
    \label{eq:10}
    \hom(S_{i,j},G)=\sum_{p,q\in\Nat}p^iq^ja_{p,q}(G).
  \end{equation}
  Moreover, $\hom(W,G)=\hom(S_{1,0},G)=\sum_{p,q\in\Nat}p a_{p,q}(G)$
  and $\hom(B,G)=\hom(S_{0,1},G)=\sum_{p,q\in\Nat}q a_{p,q}(G)$.
  Observe that $\hom(F,G)=0$ for all connected
  $F\in\TD_2\setminus(\CS\cup\{B,W\})$.

  Suppose for contradiction that there is a finite $\CF\subseteq\TD_2$
  such that for all graphs $G,G'$, if $\hom(F,G)=\hom(F,G')$ for all
  $F\in\CF$ then
  $G\equiv_2^\LC G'$. Without loss of generality we assume that all
  $F\in\CF$ are connected. We will only consider graphs $G,G'\in\CS$. Thus
  it suffices to consider $F\in\CF\cap \{S_{i,j}\mid
  i,j\in\Nat\}$. Let $m:=\max\{j\mid S_{i,j}\in \CF\}$.

  \begin{claim}
    There are vectors $\vec a=(a_1,\ldots,a_{m})$,\newline
    $\vec
    a'=(a_1',\ldots,a_{m}')\in\Nat^{m}$ such that
    \begin{eroman}
    \item $\sum_{q=1}^{m}a_q<\sum_{q=1}^{m}a_q'$,
    \item $\sum_{q=1}^{m}q^ja_q=\sum_{q=1}^{m}q^ja_q'$ for all
      $j\in[m]$. 
    \end{eroman}

    \proof
    Let $A\in\Rat^{m\times m}$ be the matrix with entries
    $a_{ij}=j^{i-1}$. Then $A$ is a Vandermonde matrix and thus has
    full rank. Therefore, the equation $A\vec x=\vec e_1$, where
    $\vec e_1=(1,0,\ldots,0)^T\in\Rat^m$, has the rational solution
    $\vec x=A^{-1}\vec e_1$. Multiplying with a positive common denominator $c$ of the entries
    of $\vec x$, we obtain an integer solution $\vec y$ to the system
    $A\vec y=c\vec e_1$. We write $\vec y=\vec a'-\vec a$ for two
    nonnegative integer vectors $\vec a'=(a'_1,\ldots,a'_n), \vec
    a=(a_1,\ldots,a_m)$. The equation $A\vec a'=c\vec e_1+A\vec a$
    yields exactly the equations in the assertion of the claim.
    \uend 
  \end{claim}

  We choose vectors $\vec a,\vec a'$ according to the claim and let
  $a_0:=\sum_{p=1}^{m}a_p'-\sum_{p=1}^{m}a_p$. Let $G$ be the
  disjoint union of $a_q$ copies of $S_{1,q}$ for
    $q\in\{0,\ldots,m\}$. Then $a_{1,q}(G)=a_q$ for $q\in
    \{0,\ldots,m\}$ and $a_{p,q}(G)=0$ for all $p\neq 1,q$ or
    $p=1,q>m$. Similarly, let $G'$ be the
  disjoint union of $a'_q$ copies of $S_{1,q}$ for
    $q\in\{1,\ldots,m\}$. Then $a_{1,q}(G')=a'_q$ for $q\in
    \{1,\ldots,m\}$ and $a_{p,q}(G')=0$ for all $p\neq 1,q$ or
    $p=1,q=0$ or $p=1,q>m$. Then for all $i\in\Nat$ we have
    \[
      \hom(S_{i,0},G)=\sum_{q=0}^{m}a_q=\sum_{q=1}^{m}a_q'=\hom(S_{i,0},G')
    \]
    and, for $j\in[m]$, 
    \[
      \hom(S_{i,j},G)=\sum_{q=1}^{m}q^ja_q=\sum_{q=1}^{m}q^ja_q'=\hom(S_{i,j},G').
    \]
    Thus $\hom(F,G)=\hom(F,G')$ for all $F\in\CF$.

    However, $G$ contains a copy of $S_{1,0}$, whereas $G'$ does
    not. Thus $G$ satisfies the $\FO_2$-sentence
    \[
      \exists x\Big(\gamma(x)\!=\!\tikz{\draw[fill=grey] (0,0) circle
        (1mm);}\wedge\exists y\big(E(x,y)\wedge \gamma(y)\!=\!\tikz{\draw[fill=white] (0,0) circle
        (1mm);}\big) \wedge\neg\exists y\big(E(x,y)\wedge \gamma(y)\!=\!\tikz{\draw[fill=black] (0,0) circle
        (1mm);}\big)\Big)
    \]
    and $G'$ does not. Hence $G$ and $G'$ are not $\FO^2$-equivalent.
    \uend
\end{example}

\subsection{Relational Structures}\label{sec:rel}
Our main result extends from graphs to arbitrary relational
structures. The definition of elimination forests and hence of tree depth
can be extended in a straightforward way. Lemma~\ref{lem:ind-td}, the
inductive characterisation of tree depth, does not generalise
directly, but can be adapted: when deleting the root, rather than
removing all tuples that contain the root from all relations, we need
to add relations of smaller arity and keep the remaining tuples after
deleting the root. A similar adaptation will be necessary in the
definition of $G\wr v$ in Section~\ref{sec:main-proof}. It needs to
be checked that Lemma~\ref{lem:1} still holds with the new
definitions---it does---, the rest of the proof goes through nearly
unchanged.

It would be interesting to work out an extension of the main theorem
to weighted graphs, yielding homomorphisms whose weight is the product
of the edges weights in its image. Such an extension would require a
suitable extension of the logic. We leave this for future work.

\subsection{Complexity}
Böker, Chen, Grohe, and Rattan~\cite{bokchegrorat19} studied the
computational complexity of homomorphism indistinguishability over
classes $\CF$ of graphs. Depending on $\CF$, they found complexities
ranging from polynomial time to undecidable. Notably,
homomorphism indistinguishability over the class of all graphs is
equivalent to isomorphism and hence decidable in quasi-polynomial time
\cite{bab16}.

It is a consequence of our main theorem that for every $k$,
homomorphism indistinguishability over $\TD_k$ is decidable in
polynomial time, or more precisely, time $n^{O(k)}$, because
$\LC_k$-equivalence is decidable in this time. Probably the
easiest way to see this is via the bijective pebble game: given
graphs $G,G'$, by induction on $\ell$ we can compute the partition of
$V(G)^{k-\ell}\cup V(G')^{k-\ell}$ such that
Duplicator wins the $\ell$-move bijective pebble game with initial
position $\vec x,\vec x'$ if and only $\vec x,\vec x'$ belong to the
same class of the partition.

We leave open the question whether homomorphism indistinguishability over
$\TD_k$ is fixed-parameter tractable when parameterised by $k$. We conjecture that it is not.

\section{Concluding Remarks}

We characterise equivalence in the counting extension of first-order
logic, parameterised by quantifier rank, in terms of homomorphism
indistinguishability over  graphs of bounded tree depth. While a
result along these lines may not be unexpected, it is
surprising that we obtain such a tight and clean correspondence between
quantifier rank and tree depth.

An interesting aspect of the correspondence between homomorphism counts
and logical equivalence is that homomorphism counts also give us a
natural notion of distance and similarity between graphs via distances between the homomorphism vectors
$\HOM_{\CF}(G)$ in suitable inner-product spaces. Through the translation
between logic and homomorphism counts, such distance measures between
graphs give us notions of ``approximate logical equivalence'' and
possibly ``approximate logical satisfiability'', which in times of
uncertain data seems very desirable and deserves further
exploration.


\begin{thebibliography}{10}

\bibitem{atsmanrob+19}
A.~Atserias, Laura Man{\v c}inska, D.E. Roberson, R.~{\v S}{\'a}mal,
  S.~Severini, and A.~Varvitsiotis.
\newblock Quantum and non-signalling graph isomorphisms.
\newblock {\em Journal of Combinatorial Theory, Series B}, 136:289--328, 2019.

\bibitem{atsman13}
A.~Atserias and E.~Maneva.
\newblock {S}herali--{A}dams relaxations and indistinguishability in counting
  logics.
\newblock {\em {SIAM} Journal on Computing}, 42(1):112--137, 2013.

\bibitem{atsoch18}
A.~Atserias and J.~Ochremiak.
\newblock Definable ellipsoid method, sums-of-squares proofs, and the
  isomorphism problem.
\newblock In {\em Proceedings of the 33rd Annual {ACM/IEEE} Symposium on Logic
  in Computer Science}, pages 66--75, 2018.

\bibitem{bab16}
L.~Babai.
\newblock Graph isomorphism in quasipolynomial time.
\newblock In {\em Proceedings of the 48th Annual {ACM} Symposium on Theory of
  Computing ({STOC}~'16)}, pages 684--697, 2016.

\bibitem{bantan16}
M.~Bannach and T.~Tantau.
\newblock Parallel multivariate meta-theorems.
\newblock In J.~Guo and D.~Hermelin, editors, {\em Proceedings of the 11th
  International Symposium on Parameterized and Exact Computation}, volume~63 of
  {\em LIPIcs}, pages 4:1--4:17. Schloss Dagstuhl - Leibniz-Zentrum f{\"u}r
  Informatik, 2016.

\bibitem{bergro15}
C.~Berkholz and M.~Grohe.
\newblock Limitations of algebraic approaches to graph isomorphism testing.
\newblock In M.M. Halld{\'{o}}rsson, K.~Iwama, N.~Kobayashi, and B.~Speckmann,
  editors, {\em Proceedings of the 42nd International Colloquium on Automata,
  Languages and Programming, Part I}, volume 9134 of {\em Lecture Notes in
  Computer Science}, pages 155--166. Springer Verlag, 2015.

\bibitem{bok18}
J.~B{\"o}ker.
\newblock Structural similarity and homomorphism counts.
\newblock Master Thesis at RWTH Aachen, 2018.

\bibitem{bokchegrorat19}
J.~B{\"o}ker, Y.~Chen, M.~Grohe, and G.~Rattan.
\newblock The complexity of homomorphism indistinguishability.
\newblock In P.~Rossmanith, P.~Heggernes, and J.-P. Katoen, editors, {\em
  Proceedings of the 44th International Symposium on Mathematical Foundations
  of Computer Science}, volume 138 of {\em Leibniz International Proceedings in
  Informatics (LIPIcs)}, pages 54:1--54:13. Schloss Dagstuhl--Leibniz-Zentrum
  fuer Informatik, 2019.

\bibitem{borchalov+06}
C.~Borgs, J.~Chayes, L.~Lov{\'a}sz, V.~S{\'o}s, B.~Szegedy, and
  K.~Vesztergombi.
\newblock Graph limits and parameter testing.
\newblock In {\em Proceedings of the 38th Annual ACM Symposium on Theory of
  Computing}, pages 261--270, 2006.

\bibitem{buldaw14}
J.~Bulian and A.~Dawar.
\newblock Graph isomorphism parameterized by elimination distance to bounded
  degree.
\newblock In M.~Cygan and P.~Heggernes, editors, {\em Proceedings of the 9th
  International Symposium on Parameterized and Exact Computation}, volume 8894
  of {\em Lecture Notes in Computer Science}, pages 135--146. Springer Verlag,
  2014.

\bibitem{caifurimm92}
J.~Cai, M.~F{\"u}rer, and N.~Immerman.
\newblock An optimal lower bound on the number of variables for graph
  identification.
\newblock {\em Combinatorica}, 12:389--410, 1992.

\bibitem{cheflu18}
Y.~Chen and J.~Flum.
\newblock Tree-depth, quantifier elimination, and quantifier rank.
\newblock In {\em Proceedings of the 33rd Annual ACM/IEEE Symposium on Logic in
  Computer Science}, pages 225--234, 2018.

\bibitem{delgrorat18}
H.~Dell, M.~Grohe, and G.~Rattan.
\newblock Lov{\'a}sz meets {W}eisfeiler and {L}eman.
\newblock In I.~Chatzigiannakis, C.~Kaklamanis, D.~Marx, and D.~Sannella,
  editors, {\em Proceedings of the 45th International Colloquium on Automata,
  Languages and Programming (Track A)}, volume 107 of {\em LIPIcs}, pages
  40:1--40:14. Schloss Dagstuhl - Leibniz-Zentrum f{\"u}r Informatik, 2018.

\bibitem{dvo10}
Z.~Dvor{\'{a}}k.
\newblock On recognizing graphs by numbers of homomorphisms.
\newblock {\em Journal of Graph Theory}, 64(4):330--342, 2010.

\bibitem{ebbflutho94}
H.-D. Ebbinghaus, J.~Flum, and W.~Thomas.
\newblock {\em Mathematical Logic}.
\newblock Springer Verlag, 2nd edition, 1994.

\bibitem{elbgrotan16}
M.~Elberfeld, M.~Grohe, and T.~Tantau.
\newblock Where first-order and monadic second-order logic coincide.
\newblock {\em ACM Transaction on Computational Logic}, 17(4), 2016.
\newblock Article No.~25.

\bibitem{elbjaktan12}
M.~Elberfeld, A.~Jakoby, and T.~Tantau.
\newblock Algorithmic meta theorems for circuit classes of constant and
  logarithmic depth.
\newblock In C.~D{\"u}rr and T.~Wilke, editors, {\em Proceedings of the 29th
  International Symposium on Theoretical Aspects of Computer Science},
  volume~14 of {\em LIPIcs}, pages 66--77. Schloss Dagstuhl - Leibniz-Zentrum
  fuer Informatik, 2012.

\bibitem{flugro04b}
J.~Flum and M.~Grohe.
\newblock The parameterized complexity of counting problems.
\newblock {\em {SIAM} Journal on Computing}, 33(4):892--922, 2004.

\bibitem{gragropagpak19}
E.~Gr{\"a}del, M.~Grohe, B.~Pago, and W.~Pakusa.
\newblock A finite-model-theoretic view on propositional proof complexity.
\newblock {\em Logical Methods in Computer Science}, 15(1):4:1--4:53, 2019.

\bibitem{groott15}
M.~Grohe and M.~Otto.
\newblock Pebble games and linear equations.
\newblock {\em Journal of Symbolic Logic}, 80(3):797--844, 2015.

\bibitem{hel96}
L.~Hella.
\newblock Logical hierarchies in {PTIME}.
\newblock {\em Information and Computation}, 129:1--19, 1996.

\bibitem{imm87a}
N.~Immerman.
\newblock Expressibility as a complexity measure: results and directions.
\newblock In {\em Proceedings of the 2nd IEEE Symposium on Structure in
  Complexity Theory}, pages 194--202, 1987.

\bibitem{immlan90}
N.~Immerman and E.~Lander.
\newblock Describing graphs: A first-order approach to graph canonization.
\newblock In A.~Selman, editor, {\em Complexity theory retrospective}, pages
  59--81. Springer-Verlag, 1990.

\bibitem{krijohmor19}
N.M. Kriege, F.D Johansson, and C.~Morris.
\newblock A survey on graph kernels.
\newblock {\em ArXiv}, arXiv:1903.11835 [cs.LG], 2019.

\bibitem{kusschwe17}
D.~Kuske and N.~Schweikardt.
\newblock First-order logic with counting.
\newblock In {\em Proceedings of the 32nd ACM-IEEE Symposium on Logic in
  Computer Science}, 2017.

\bibitem{kusschwe18}
D.~Kuske and N.~Schweikardt.
\newblock Gaifman normal forms for counting extensions of first-order logic.
\newblock In I.~Chatzigiannakis, C.~Kaklamanis, D.~Marx, and D.~Sannella,
  editors, {\em Proceedings of the 45th International Colloquium on Automata,
  Languages, and Programming}, volume 107 of {\em LIPIcs}, pages 133:1--133:14.
  Schloss Dagstuhl - Leibniz-Zentrum f{\"u}r Informatik, 2018.

\bibitem{lib98}
L.~Libkin.
\newblock On counting and local properties.
\newblock In {\em Proceedings of the 13th IEEE Symposium on Logic in Computer
  Science}, pages 501--512, 1998.

\bibitem{lib99}
L.~Libkin.
\newblock Logics with counting, auxiliary relations, and lower bounds for
  invariant queries.
\newblock In {\em Proceedings of the 14th IEEE Symposium on Logic in Computer
  Science}, 1999.

\bibitem{lib00a}
L.~Libkin.
\newblock Logics capturing local properties.
\newblock In Horst Reichel and Sophie Tison, editors, {\em 17th Annual
  Symposium on Theoretical Aspects of Computer Science}, volume 1770 of {\em
  Lecture Notes in Computer Science}, pages 217--229. Springer-Verlag, 2000.

\bibitem{lov67}
L.~Lov{\'a}sz.
\newblock Operations with structures.
\newblock {\em Acta Mathematica Hungarica}, 18:321--328, 1967.

\bibitem{lov12}
L.~Lov{\'a}sz.
\newblock {\em Large Networks and Graph Limits}.
\newblock American Mathematical Society, 2012.

\bibitem{lovsze06}
L.~Lov{\'a}sz and B.~Szegedy.
\newblock Limits of dense graph sequences.
\newblock {\em Journal of Combinatorial Theory, Series B}, 96(6):933--957,
  2006.

\bibitem{maent19}
T.~Maehara and H.~NT.
\newblock A simple proof of the universality of invariant/equivariant graph
  neural networks.
\newblock {\em ArXiv (CoRR)}, arXiv:1910.03802v1 [cs.LG], 2019.

\bibitem{mal14}
P.~Malkin.
\newblock Sherali--adams relaxations of graph isomorphism polytopes.
\newblock {\em Discrete Optimization}, 12:73--97, 2014.

\bibitem{manrob19}
L.~Man{\v c}inska and D.E. Roberson.
\newblock Quantum isomorphism is equivalent to equality of homomorphism counts
  from planar graphs.
\newblock {\em ArXiv}, arXiv:1910.06958v2 [quant-ph], 2019.

\bibitem{morritfey+19}
C.~Morris, M.~Ritzert, M.~Fey, W.~Hamilton, J.E. Lenssen, G.~Rattan, and
  M.~Grohe.
\newblock Weisfeiler and leman go neural: Higher-order graph neural networks.
\newblock In {\em Proceedings of the 33rd AAAI Conference on Artificial
  Intelligence}, volume 4602-4609. {AAAI} Press, 2019.

\bibitem{nesoss06}
J.~Ne{\v s}et{\v r}il and P.~Ossona de~Mendez.
\newblock Linear time low tree-width partitions and algorithmic consequences.
\newblock In {\em Proceedings of the 38th ACM Symposium on Theory of
  Computing}, pages 391--400, 2006.

\bibitem{nesoss12}
J.~{Ne\v{s}et{\v r}il} and P.~{Ossona de Mendez}.
\newblock {\em Sparsity}.
\newblock Springer-Verlag, 2012.

\bibitem{schwe19}
N.~Schweikardt.
\newblock Local normal forms and their use in algorithmic meta theorems
  (invited talk).
\newblock In {\em Proceedings of the 34th Annual {ACM/IEEE} Symposium on Logic
  in Computer Science}, pages 1--3, 2019.

\bibitem{sheschlee+11}
N.~Shervashidze, P.~Schweitzer, E.J. van Leeuwen, K.~Mehlhorn, and K.M.
  Borgwardt.
\newblock Weisfeiler-{L}ehman graph kernels.
\newblock {\em Journal of Machine Learning Research}, 12:2539--2561, 2011.

\bibitem{tod91}
S.~Toda.
\newblock {PP} is as hard as the polynomial-time hierarchy.
\newblock {\em {SIAM} Journal on Computing}, 20(5):865--877, 1991.

\end{thebibliography}

\end{document}